\documentclass[letterpaper,12pt]{amsart}
\usepackage[margin=1in]{geometry}
\usepackage{ragged2e}

\usepackage{tikz} 
\usetikzlibrary{calc} 
\usetikzlibrary{arrows.meta} 

\usepackage{latexsym,array,delarray,amsthm,amssymb,epsfig,tabu}
\usepackage{caption}
\usepackage{diagbox}
\usepackage{hyperref} 
\usepackage[capitalize]{cleveref} 
\usepackage{caption}

\usepackage{thmtools,thm-restate} 
\newtheorem{theorem}{Theorem}[section]  
\newtheorem{lemma}[theorem]{Lemma} 
\newtheorem{corollary}[theorem]{Corollary} 

\theoremstyle{definition} 
\newtheorem{example}[theorem]{Example} 

\theoremstyle{remark} 
\newtheorem{remark}{Remark} 
\newtheorem{observation}{Observation} 

\crefname{ineq}{Ineq.}{Ineqs.} 
\Crefname{ineq}{Inequality}{Inequalities} 
\crefalias{ineq}{equation} 
\creflabelformat{ineq}{(#2#1#3)}  
\makeatletter 
\AddToHook{cmd/appendix/before}{\def\cref@section@alias{appendix}} 
\makeatother 

\newcounter{claimcount}
\newenvironment{claim}{\refstepcounter{claimcount}\par\addvspace{\bigskipamount}\noindent\textbf{Claim \arabic{claimcount}:}}{}
\usepackage{etoolbox}
\AtBeginEnvironment{proof}{\setcounter{claimcount}{0}}
\newenvironment{claimproof}{\par\addvspace{\medskipamount}\noindent\textit{Proof of Claim  \arabic{claimcount}.}}{\hfill\ensuremath{{\scriptstyle\square}} \tiny{Claim}

  \medskip}
\crefname{claimcount}{Claim}{Claims}

\newcommand{\im}{\text{im}} 

\usepackage{natbib}

\usepackage{lineno}

\usepackage{setspace}

\begin{document}
\title[Identifying Triangles in Phylogenetic Networks]{Semialgebraic Conditions for Identifying Triangles in Phylogenetic Networks}

\author[]{Bryan Currie}
\address{New Jersey Institute of Technology}
\email{bc479@njit.edu}

\author[]{Aviva K. Englander}
\address{University of Wisconsin Madison}
\email{akenglander@wisc.edu}

\author[]{Jose A. Esparza-Lozano}
\address{University of Hawai`i at M\={a}noa}
\email{joseael@hawaii.edu}

\author[]{Elizabeth Gross}
\address{University of Hawai`i at M\={a}noa}
\email{egross@hawaii.edu}

\author[]{Max Hill}
\address{University of Hawai`i at M\={a}noa}
\email{max6@hawaii.edu}

\author[]{Colby Long}
\address{The College of Wooster}
\email{clong@wooster.edu}

\author[]{Devon Olds}
\address{North Carolina State University}
\email{dolds@ncsu.edu}

\author[]{Kawika O'Connor}
\address{University of Hawai`i at M\={a}noa}
\email{oconnor3@hawaii.edu}

\author[]{Udani Ranasinghe}
\address{University of Hawai`i at M\={a}noa}
\email{udanir@hawaii.edu}

\author[]{Christin Sum}
\address{University of Hawai`i at M\={a}noa}
\email{csum@hawaii.edu}

\begin{abstract}
  An important consideration for a model-based method of phylogenetic
  network inference is the identifiability of the network parameter of the
  model. A recurring theme in previous works exploring this issue is that it
  is often difficult to identify the orientation of edges in a triangle of the
  network. In fact, it has been shown that for some models it is impossible to
  determine the orientation of triangle edges utilizing the standard algebraic
  technique of phylogenetic invariants. In this work, we consider one such
  model with a Jukes-Cantor site-substitution process and no coalescence. We
  give a complete semialgebraic description of three, 3-leaf Jukes-Cantor
  phylogenetic network models with embedded triangles. By describing these
  base cases, we resolve several questions about the identifiability of
  networks with embedded triangles. We show that for any pair of models, the
  intersection and set differences of the models are full-dimensional regions
  of the space of site-pattern probability distributions. Thus, despite being
  algebraically indistinguishable, these network models are not identical, nor
  are they identifiable (or generically identifiable). Our results also yield
  a straightforward biological interpretation--that the signal from a
  hybridization event may be immediately detectable but decays over time until it
  is impossible to identify the orientation of edges in the triangle of a
  network.
\end{abstract}

\maketitle

\section{Introduction}

A phylogenetic network model represents non-tree-like evolution along a
directed acyclic graph. In such a model, network nodes of in-degree two
correspond to reticulation events, such as hybridization or lateral gene
transfer, where a certain proportion of the information at the node is assumed
to have been inherited along each of the incoming edges. Recognizing that
reticulation events are likely common in the evolutionary history of many
species, there has been increasing interest in the theory of phylogenetic
network models of DNA sequence evolution and their application to phylogenetic
inference \citep{barley2022evolutionary, Cui2013-ps, Gambette2017-kf,
hibbins2022phylogenomic, mallet2016reticulated, Pardi2015-ix, lobeloids,
solislemus2016snaq, vanderHeijden2025, Zhang2018-zk}.

One important consideration for developing model-based methods of phylogenetic
inference is the identifiability of the model parameters. In general, a model
parameter is identifiable if it can be uniquely recovered from some output of
the model. In the ideal case for phylogenetic inference, the observed data
will perfectly fit a phylogenetic model and the corresponding model parameters
are inferred. However, for this inference to be consistent, there must be a
unique set of parameters corresponding to the data, that is, the parameters
must be identifiable. As we consider the specific setting of networks, the
same holds: to develop a consistent model-based method of network inference,
the network parameter of the model must be identifiable. Several works have
studied identifiability in phylogenetic network models with and without a
coalescent process \citep{allman2025beyond, allman2024identifiability,
allman2022identifiability, banos2019identifying, englander2025identifiability,
gross2018distinguishing,gross2021distinguishing, holtgrefe2025distinguishing,
rhodes2025identifying,xu2023identifiability}.

One recurring theme in these works is that while many network features are
identifiable from data, it is often not possible to identify network
\emph{triangles} or to determine the orientation of edges in a triangle of the
network. For example, the results of \cite{gross2018distinguishing} establish
only the identifiability of \emph{triangle-free} networks, and the standard
methods of algebraic statistics employed therein are unable to distinguish
networks such as those shown in \Cref{fig:3cycles}. Likewise,
\cite{barton2022statistical,englander2025identifiability,gross2021distinguishing}
and \cite{holtgrefe2025distinguishing} also consider only triangle-free
networks due to the inadequacy of algebraic methods for distinguishing
triangles. It is thought that hybridization is negatively correlated with
genetic distance between species \citep{mallet2005hybridization}. Since
reticulation is expected to be most common among closely related species,
triangles are likely an important part of the evolutionary history of many
sets of taxa, making the gap in the current results especially unfortunate.
One notable exception to this trend are the results of
\cite{allman2024identifiability}, who showed that under the network multispecies
coalescent model, 3-cycles can sometimes be detected from gene tree quartet
concordance factors.

The results of \cite{englander2025identifiability,gross2018distinguishing} and
\cite{gross2021distinguishing} apply to DNA sequence models without a
coalescent process. Here, each choice of parameters produces a probability
distribution on the possible DNA site-patterns that may be observed in the
aligned DNA sequences of the species under consideration. The algebraic
approach used to prove the identifiability of the network parameter is based
on finding \emph{phylogenetic invariants} for each network. These are
polynomial relationships that are always satisfied by the site-pattern
probability distributions coming from a model on the network. The key idea is
to show that every pair of networks is \emph{algebraically distinguishable},
that is, there is a polynomial invariant for each that is not an invariant for
the other. Thus, given a site-pattern probability distribution from a network
model, one can evaluate the invariants at the distribution to determine the
network parameter that produced the data. This is a standard technique in
algebraic statistics and phylogenetic invariants have been used to establish
many identifiability results \citep[Chapter 15,
p.335-370]{sullivant2023algebraic}.
Often, the approach is to first find phylogenetic invariants for trees or
networks with just a few taxa using computational tools, and then, to use
combinatorial arguments and restrictions to subsets of taxa to show that trees
or networks with an arbitrary number of leaves are distinguishable.

As noted above, despite the power of the algebraic approach, for certain
models these methods are simply insufficient to establish the identifiability
of networks with embedded triangles. For example, under the Jukes-Cantor
model, there are no phylogenetic invariants for the three networks shown in
\Cref{fig:3cycles} and so they are not algebraically distinguishable.

\begin{figure*}
  \centering \includegraphics[width=13cm]{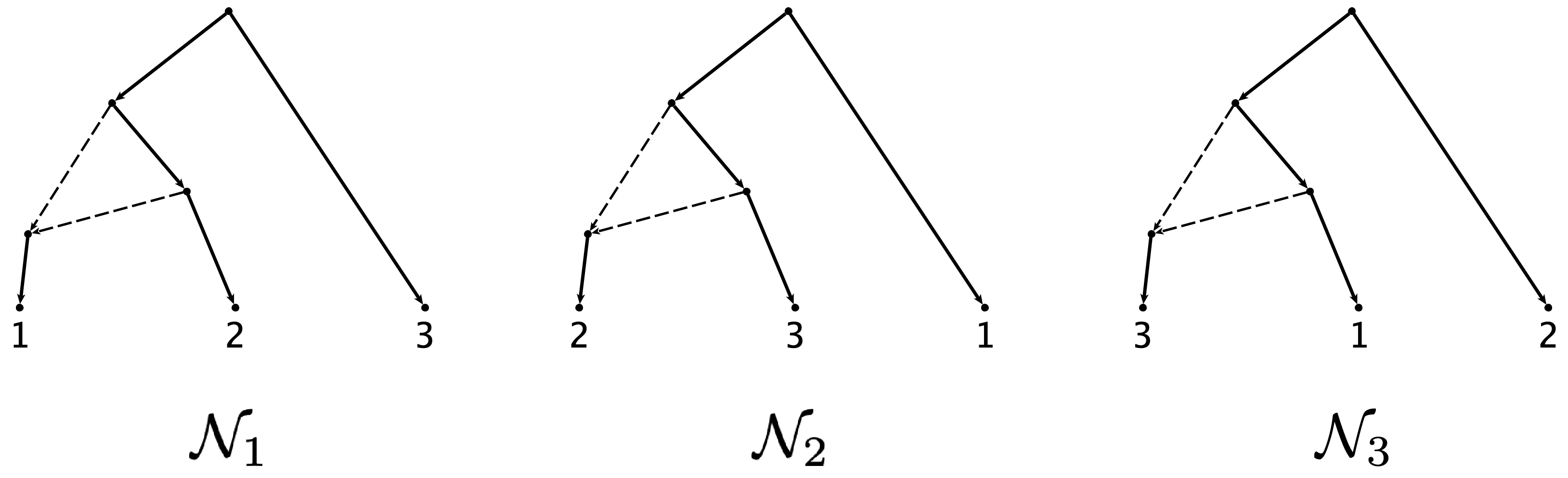}
  \caption{The three 3-leaf level-one triangle networks, each rooted on one of
    the non-hybrid leaf edges.}
  \label{fig:3cycles}
\end{figure*}

For those familiar with considering phylogenetic models from an algebraic
statistics perspective, the fact that there are no invariants for these
networks means that each of the models is a full-dimensional subset of the
probability simplex and its Zariski closure fills the entire space. However,
although their Zariski closures are equal, the example below gives some
intuition for why we may not expect the models themselves to be equal.

\begin{example}\label{ex:maxexample}
  Consider the site-pattern probability distribution generated by the
  Jukes-Cantor model on $\mathcal{N}_1$, the rooted 3-leaf network from
  \Cref{fig:3cycles}. \Cref{fig:maxexample} (left) shows this network with
  branch lengths given in expected number of mutations per site. The dotted
  reticulation edges are length $\epsilon > 0$ and represent a near
  instantaneous hybridization event. For this example, we assume half of the
  genetic information is inherited along each edge.

  If $\ell$ is sufficiently large, then the sites observed at leaves 2 and 3
  are nearly independent. If it is also the case that $s$ is very small, then
  the resulting site pattern probability distribution from the model
  represents essentially a coin flip as to whether the site observed at leaf 1
  is identical to the site at leaf 2 or at leaf 3. It seems impossible for
  such a distribution to result from any choice of branch lengths and
  reticulation parameter for networks $\mathcal{N}_{2}$ or $\mathcal{N}_{3}$.
  Indeed, by applying the results of
  \cref{thm:parameter-conditions-for-belonging-to-m1-intersect-m2}, we can
  show that for fixed $\ell$, when $s$ is below a certain threshold the
  resulting site-pattern probability distribution does not belong to the
  models on networks $\mathcal N_2$ and $\mathcal N_3$. The biological
  interpretation of this result is rather straightforward--it is generally
  possible to identify the precise nature of a hybridization event immediately
  after it occurs, but this signal decays until it eventually becomes
  impossible to determine the orientation of edges in the triangle of the
  network.

  \Cref{fig:maxexample} (right) shows this relationship explicitly for the
  limiting case as $\epsilon \to 0$. In the figure, points $(s, \ell)$ in the
  shaded region result in distributions that belong exclusively to the model
  on $\mathcal N_1$. For a fixed $\ell$, as $s$ increases we eventually
  approach the threshold of identifiability beyond which it is impossible to
  determine the orientation of edges in the triangle. As $\ell\to\infty$, the
  boundary threshold for $s$ approaches the asymptote
  $s=\ell-\frac{3}{4}\log(4)$, and as $\ell\to 0$, it approaches
  $\ell = (1+\sqrt{5})s$. The second asymptote implies that when $\ell$ is
  small (the biologically relevant case), distinguishing the hybrid node
  requires $s$ to be less than approximately
  $\ell/(1+\sqrt{5})\approx \ell/3.2$.
\end{example}

If the models are not equal, it is possible that there are polynomial
\emph{inequalities} that hold for each model that would allow us to
distinguish them. Such inequalities have also been used to establish
identifiability for phylogenetic models \citep{allman2024identifiability,
  englander2025identifiability}.

In this work, we consider the Jukes-Cantor model on the 3-leaf phylogenetic
networks with embedded triangles shown in \Cref{fig:3cycles}. By finding
polynomial inequalities for these base case models, we are able to resolve
several questions about the identifiability of networks with embedded
triangles. For example, our results show that in general these networks are
not identifiable (or generically identifiable), however, we also show that
none of the 3-leaf triangle network models are identical. That is, each
network model contains site-pattern probability distributions that do not
belong to the other two models.

The remainder of our paper proceeds as follows. In \textit{Materials and
  Methods}, we describe the phylogenetic network models and their
parameterizations in the Fourier coordinates. In \textit{Results}, we first
describe the identifiability of numerical parameters within a single model,
and then give a complete semialgebraic description of the network model for
$\mathcal{M}_{1}$, the model on $\mathcal{N}_{1}$. In the subsection
\textit{Model Intersections and Implications for Identifiability}, we apply
the above results to explore the intersections of all three Jukes-Cantor
3-leaf triangle network models, including giving necessary and sufficient
conditions in terms of the numerical parameters for the resulting site-pattern
probability distribution to lie in the intersection of two or more models. We
then further explore the size of model intersections. Finally, we present two
applications to biological data and discuss the implications of these results
for the identifiability and practical inference of phylogenetic networks.

\begin{figure*}[]
  \centering

  \begin{minipage}[c]{0.40\textwidth}
    \centering
    \begin{tikzpicture}[
      baseline=(current bounding box.center),
      line cap=round,
      line join=round,
      line width=1pt,
      >={Stealth[length=5pt,width=4pt]}
      ]

      \coordinate (R)  at (0,4.8);
      \coordinate (A)  at (-1.7,1.4);
      \coordinate (B)  at (0,1.4);
      \coordinate (C)  at (1.7,1.4);

      \coordinate (L2) at (-2,-.5);
      \coordinate (L1) at (0,-.5);
      \coordinate (L3) at (2,-.5);

      \draw[->, shorten >=1pt, ] (R) -- (A);
      \draw[->, shorten >=1pt, ] (R) -- (C);

      \draw[->, dashed, shorten >=1pt, ] (A) -- (B);
      \draw[->, dashed, shorten >=1pt, ] (C) -- (B);

      \draw[->, shorten >=1pt, ] (A) -- (L2);
      \draw[->, shorten >=1pt, ] (B) -- (L1);
      \draw[->, shorten >=1pt, ] (C) -- (L3);

      \foreach \p in {R,A,B,C,L1,L2,L3}
      \fill (\p) circle (1.1pt);

      \node[below, font=\sffamily] at (L2) {2};
      \node[below, font=\sffamily] at (L1) {1};
      \node[below, font=\sffamily] at (L3) {3};

      \node[left]  at ($(R)!0.52!(A)$) {$\ell/2$};
      \node[right] at ($(R)!0.52!(C)$) {$\ell/2$};

      \node[left]  at ($(A)!0.5!(L2)$) {$s$};
      \node[left]  at ($(B)!0.5!(L1)$) {$s$};
      \node[right] at ($(C)!0.5!(L3)$) {$s$};
    \end{tikzpicture}
  \end{minipage}
  \quad
  \begin{minipage}[c]{0.48\textwidth}
    \centering
    \begin{tikzpicture}[baseline=(current bounding box.center)]
      \node[inner sep=0] (image) {\includegraphics[height=2.2in,width=2.2in]{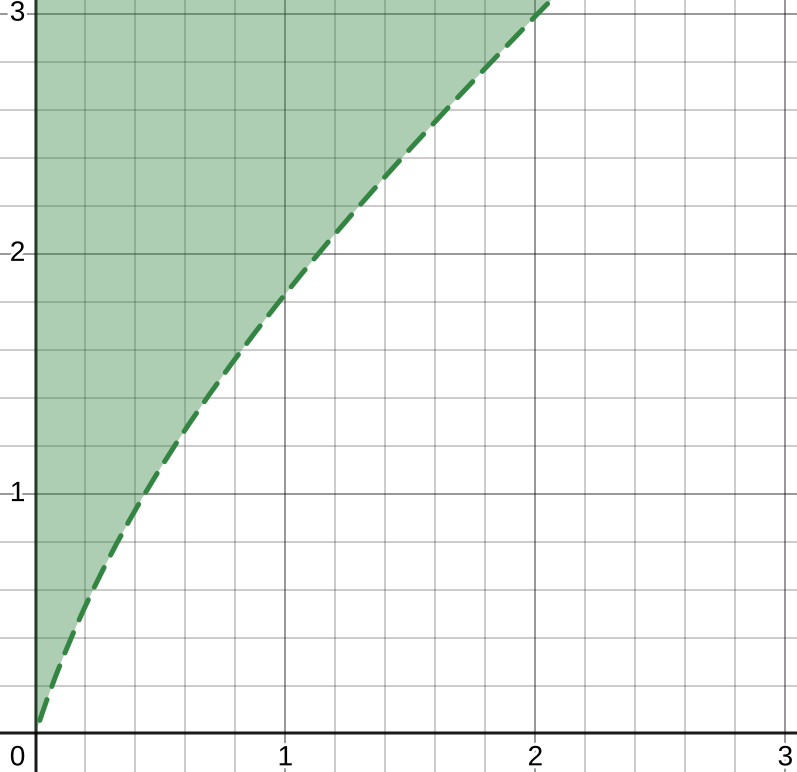}
      };

      \node[font=\sffamily] at ($(image.south)+(0,-11pt)$) {branch length $s$};

      \node[rotate=90,font=\sffamily] at ($(image.west)+(-11pt,0)$) {branch length $\ell$};
    \end{tikzpicture}
  \end{minipage}

  \caption{Left: The network $\mathcal{N}_1$ from \Cref{fig:3cycles}, with
    branch lengths given in expected number of mutations per site. Right: The
    shaded region represents the choices of branch lengths (in terms of $s$
    and $\ell$) for which the resulting site-pattern probability distributions
    belong exclusively to the model on $\mathcal{N}_{1}$.}
  \label{fig:maxexample}
\end{figure*}
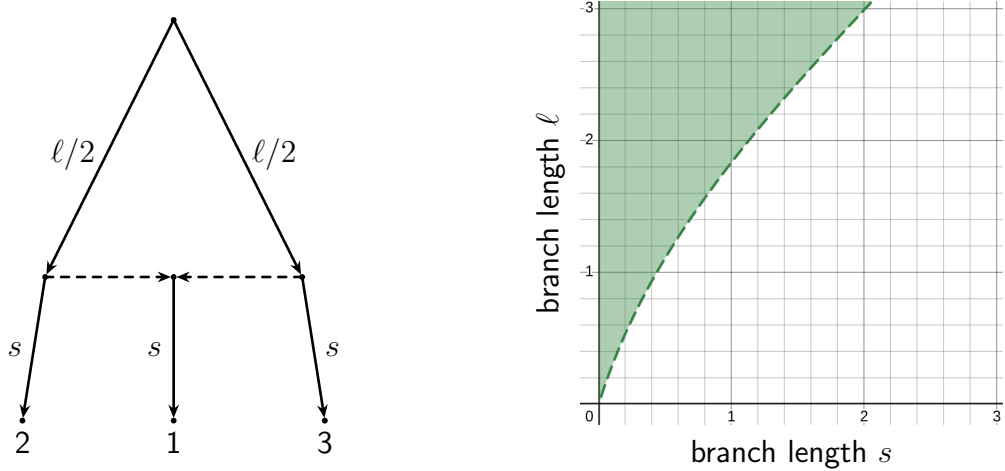

\section{Materials and Methods}
\label{sec:materials-and-methods}

A distribution in a phylogenetic network model of DNA sequence evolution can be viewed as a mixture of tree distributions. Thus we begin this section by describing tree-based Markov models of DNA sequence evolution. 

\subsection{Phylogenetic Tree Models} 
\label{sec:phylogenetic-tree-models}
Let $\mathcal T$ be a rooted $n$-leaf binary phylogenetic tree with leaves labeled by the taxa in $[n]=\{1, \ldots, n\}$ and root vertex $\rho$. To each vertex $v$ of $\mathcal{T}$, we associate a random variable $X_v$ with state space the set of four DNA bases  $\{\mathtt{A},\mathtt{G},\mathtt{C},\mathtt{T}\}$. The state of this random variable represents the nucleotide at the site being modeled in the DNA sequence of the taxon represented by the vertex $v$. 

The \emph{root distribution} of the model, $\boldsymbol{\pi}$, specifies the probability of observing each of the four DNA bases at the root of the tree. We associate to each edge of $\mathcal T$ a $4 \times 4$ Markov transition matrix with rows and columns indexed by the four DNA bases. The entries of these matrices describe the rates of mutation along each edge, that is, for edge $e = uv$ of $\mathcal T$, $M^e_{ij} = P(X_v = j \mid X_u = i).$ We refer to the root distribution and the entries of the transition matrices as the \emph{numerical parameters} of the model. 

Given an assignment of states to the vertices of $\mathcal T$,
\begin{equation*}
  \phi: V(\mathcal T) \rightarrow \{\mathtt{A}, \mathtt{G}, \mathtt{C}, \mathtt{T}\},
\end{equation*}
we can compute the probability of observing this particular assignment of states using the root distribution and the transition matrices as

\begin{align*}
  p(\phi)=\boldsymbol{\pi}_{\phi(\rho)}\prod_{e\in E(\mathcal T)}M_{\phi(u),\phi(v)}^e.
\end{align*}

To compute the probability of observing a particular \emph{site-pattern} $(X_1,\ldots,X_n)$ in the aligned DNA sequences of the taxa at the leaves, we compute the joint distribution of the leaves by marginalizing over all possible assignments of states to the interior vertices of $\mathcal T$. More formally, if $\phi|_{\mathcal L}$ is the $n$-tuple of states that $\phi$ assigns to the leaves $\mathcal L$ of $\mathcal T$, then the probability of observing the site-pattern $\omega = (i_1, i_2, \ldots, i_n)$ is

\begin{align}
  \label{eq:site-pattern-probability-formula}
  p_\omega 
  &= \sum_{\phi\,:\,\phi|_{\mathcal L} =\omega}p(\phi) 
  = \sum_{\phi\,:\,\phi|_{\mathcal L}=\omega}
    \pi_{\phi(\rho)}
    \prod_{e\in E(\mathcal T)}M_{\phi(u),\phi(v)}^e.
\end{align}

In this work, we will assume that the transition matrices come from a continuous-time Jukes-Cantor model of DNA sequence evolution. Thus, the root distribution is uniform, and each transition matrix $M^e$ has the form

\begin{equation} 
  \label{eq:transition-matrix}
  M^e_{ij} =
  \begin{cases}
    \frac{1}{4} + \frac{3}{4}e^{-4 t /3}  &\text{if } i = j\\
    \frac{1}{4} - \frac{1}{4}e^{-4 t/3 } &\text{if } i \neq j
  \end{cases},
\end{equation}
where $t$ is the branch length in expected number of mutations per site
\citep{jukes1969evolution}; see also \citep[Chapter~8, pp.194-197]{semple2003phylogenetics}.

For a fixed tree $\mathcal T$, the phylogenetic model defines a map
\begin{align} 
  \label{eq:simple-param}
  \psi_{\mathcal T}:\Theta_{\mathcal T}  \to \Delta^{4^n-1}
\end{align}
from the numerical parameter space $\Theta_{\mathcal T}$ for the model to the set of site-pattern probability distributions.
The image of this map, 
$\mathcal M_\mathcal T = \im(\psi_\mathcal T)$, is the \emph{phylogenetic model associated to $\mathcal T$}. 

\subsection{Phylogenetic Network Models} 
\label{sec:phylogenetic-network-models}
To define a site-substitution model on an $n$-leaf rooted binary phylogenetic
network $\mathcal N$, we first specify a root distribution and then associate
a transition matrix to each edge of $\mathcal{N}$ just as for a phylogenetic
tree model. Let $w_1, \ldots, w_m$ be the \emph{reticulation vertices} of
$\mathcal N$, the nodes of in-degree two, and let $e^0_i$ and $e^1_i$ be the
edges directed into $w_i$.

To obtain a tree from the network, for each $i\in [m]$, we independently
delete $e^1_i$ with probability $\delta_i \in (0,1)$, and otherwise, we delete
$e^0_i$. Intuitively, the parameter $\delta_i$ corresponds to the probability
that the particular site being modeled was inherited along edge $e^0_i$. After
deleting the $m$ edges, the result is a rooted $n$-leaf phylogenetic tree.
Each of these trees, along with its transition matrices inherited from
$\mathcal N$, gives rise to a distribution in a phylogenetic tree model, and
we can take a convex combination of these distributions to get a site-pattern
probability distribution from the network. For the networks we consider in
this paper, there is only a single reticulation node, and so for each network
we obtain a map
\begin{gather*}
  \label{eq:network-param}
  \psi_\mathcal{N}:
  \Theta_{\mathcal{N}} \times (0,1) 
  \to
  \Delta^{4^n - 1}, \\ 
  (\theta,\delta) \mapsto 
  \delta\psi_{\mathcal{T}_1}(\theta) + 
  ( 1 - \delta)\psi_{\mathcal{T}_2}(\theta).
\end{gather*}
In the expression above, $\Theta_\mathcal{N} \times (0,1)$ represents the
numerical parameter space of the model, which includes the root distribution,
transition matrices, and the single reticulation edge parameter $\delta$
chosen between $0$ and $1$. The trees $\mathcal {T}_1$ and $\mathcal {T}_2$
are the embedded trees of the network obtained by independently deleting
$e^1_1$ and $e^0_1$ respectively. Just as for a tree, we define the image of
this map, $\mathcal M_{\mathcal N} = \im(\psi_{\mathcal N})$, to be the
\emph{phylogenetic model associated to the network $\mathcal N$}. Note that as
in \citep{englander2025identifiability} we do not allow $\delta$ to be either 0
or 1. Thus, the models of the embedded trees are not contained in the network
model.

Though these models are naturally defined in terms of rooted networks, because
the Jukes-Cantor model is time-reversible, the location of the root in the
network is unidentifiable \citep{gross2021distinguishing}. Thus, we can
determine the parameterization of the phylogenetic model associated to the
network $\mathcal N$ from the semi-directed network obtained from
$\mathcal{N}$ by suppressing the root vertex and undirecting all edges other
than the reticulation edges. For example, the semi-directed versions of
$\mathcal{N}_1,\mathcal{N}_2,$ and $\mathcal{N}_3$ from \Cref{fig:3cycles} are
depicted in \Cref{fig:models-N1-N2-N3-with-labels}. For these three networks,
we denote the corresponding phylogenetic models by
$\mathcal{M}_{1},\mathcal{M}_{2},$ and $\mathcal{M}_{3}$.

\subsection{The Fourier Coordinates}
\label{sec:Fourier}

To compute the joint distribution at the leaves of a tree in a phylogenetic
model, we sum over all possible states of the internal vertices. For one of
the 3-leaf triangle networks of \Cref{fig:3cycles}, after we delete a
reticulation edge, the resulting tree has five edges, each with their own
transition matrix, and three internal vertices. Thus, following
\cref{eq:site-pattern-probability-formula}, each coordinate of the map
$\psi_{\mathcal T_i}$ is parameterized by a degree 5 polynomial with
$4^3 = 64$ terms. As a convex sum of $\psi_{\mathcal T_1}$ and
$\psi_{\mathcal T_2}$, the parameterization of the network is even more
complex, making it difficult to perform algebraic computations or to obtain a
semialgebraic description of the model. Thus, we will apply a linear change of
coordinates called the discrete Fourier transform that is defined for the
Jukes-Cantor model as well as a broader class of models known as
\emph{group-based models} \citep{evans1993invariants}. In these new
coordinates, the phylogenetic tree models are parameterized by monomials,
which makes computations much more tractable. The Jukes-Cantor model is a
group-based model, and we give a brief outline here of how to obtain the
parameterization of the network in these new coordinates. It is not necessary
to state all of the definitions surrounding group-based models and the
discrete Fourier transform in order to describe the parameterization, so we do
not delve into the details here and refer the interested reader instead to
\cite{SS05} and \cite[Chapter~15,
pp.~335-370]{sullivant2023algebraic}.

For what follows, we will assume that the tree $\mathcal T$ is
\emph{unrooted}. We identify the state space $\{\mathtt{A}, \mathtt{G}, \mathtt{C}, \mathtt{T}\}$ with the Klein
four-group $\mathbb{Z}_2\times\mathbb{Z}_2$ as follows: $\mathtt{A}=(0,0)$, $\mathtt{G}=(1,0)$,
$\mathtt{C}=(0,1)$, $\mathtt{T}=(1,1)$. We associate to each edge $e_i\in E(\mathcal{T})$ the
four eigenvalues of its transition matrix $M^{e_i}$, which we refer to as the
\emph{Fourier parameters}, denoted $a_{\mathtt{A}}^i$, $a_{\mathtt{G}}^i$, $a_{\mathtt{C}}^i$ and $a_{\mathtt{T}}^i$. The
largest of these, $a_{\mathtt{A}}^i$, is always 1 since $M^{e_i}$ is a stochastic matrix.
Due to symmetries of the Jukes-Cantor model, the remaining three Fourier
parameters satisfy $a^i_{\mathtt{C}}= a^i_{\mathtt{G}}= a^i_{\mathtt{T}} = e^{-4t_i/3}$, where $t_i$ is the
branch length of $e_i$ in expected number of mutations per site. Hence there
is really only a single Fourier parameter $a_i$ for each edge $e_i$, namely,
\begin{equation}
  \label{eq:fourier-parameter-to-branch-length}
  a_{i} = e^{-4t_{i}/3}.
\end{equation}
For this paper, we will assume that each branch length $t_i$ is finite and
positive. Notice that since $t_i> 0$, the corresponding Fourier parameter
$a_i$ lies in the interval $(0, 1)$. The parameter $a_i$ can be regarded as a
measure of the nucleotide transmission fidelity along its edge, a quantity
which decreases to $0$ as the branch length $t_i$ increases.

The Fourier transform extends to a transformation of the coordinate space,
from probability coordinates to what we will call \emph{$q$-coordinates}. The
\emph{Fourier parameterization}, which gives us the value of the
$q$-coordinate $q_\omega$ for the site-pattern
$\omega=(g_1,g_2,\dots,g_n)\in \left\{\mathtt{A},\mathtt{G},\mathtt{C},\mathtt{T}\right\}^n$, is defined as
{\small
\begin{equation} 
  \label{eq:fourier-parameterization}
  q_\omega
  =
  \begin{cases}
    \prod\limits_{\substack{e_i\in E(\mathcal{T})\\ 
    e_i\text{ induces the split $A|B$}}} a_{\sum_{j\in A}g_j}^i 
    &\text{if }\sum\limits_{j=1}^ng_j=0,\\
    0 &\text{otherwise}
  \end{cases}
\end{equation}}
where addition is in the group $\mathbb{Z}_2\times\mathbb{Z}_2$. 

These $q$-coordinates are not interpretable as probabilities; however, there
is an inverse Fourier transform that allows us to convert back and forth
between the $q$-coordinates and the probability coordinates (see
\cref{app:fourier-transform}). Thus, a semialgebraic description of a model in
terms of $q$-coordinates can be transformed to obtain a semialgebraic
description of the model in probability coordinates (and vice versa). For the
rest of this paper we will work in the Fourier coordinates, and for
simplicity, we will use $\mathcal{M}_{1}, \mathcal{M}_{2},$ and
$\mathcal{M}_{3}$ for the transformed models in the space of $q$-coordinates.

\subsection{Parameterization of the 3-Leaf Jukes-Cantor Networks}
\label{sec:parameterization}

In this section, we give the parameterization of the model
$\mathcal M_1$ in Fourier coordinates using the semi-directed
version of $\mathcal{N}_1$ with edges labeled as in
\Cref{fig:models-N1-N2-N3-with-labels}. From this parameterization,
one can easily obtain parameterizations of $\mathcal M_2$ and
$\mathcal M_3$ by permuting coordinates, as described below.

\begin{figure}[h!]
  \centering
  \begin{tikzpicture}[scale=1, transform shape,
    >={Stealth[length=4pt,width=3pt]}, line cap=round, line join=round]
    \coordinate (A) at (0,0); 
    \coordinate (B) at (2.1,0.15); 
    \coordinate (C) at (1.15,1.45); 
    \coordinate (L1) at (-0.35,-1.2); 
    \coordinate (L2) at (2.45,-1.2); 
    \coordinate (L3) at (1.15,2.6); 
    \draw[-, thick] (A) -- (L1); \draw[-, thick] (B) -- (L2); \draw[-, thick]
    (C) -- (L3); \draw[thick] (C) -- (B); \draw[->, shorten >=1pt, shorten
    <=1pt, thick, dashed] (B) -- (A); \draw[->, shorten >=1pt, shorten <=1pt,
    thick, dashed] (C) -- (A);
    \filldraw[black] (A) circle (1pt); \filldraw[black] (B) circle (1pt);
    \filldraw[black] (C) circle (1pt); \filldraw[black] (L1) circle (1pt);
    \filldraw[black] (L2) circle (1pt); \filldraw[black] (L3) circle (1pt);
    \node[below] at (L1) {1}; \node[below] at (L2) {2}; \node[above] at (L3)
    {3};
    \node[left] at ($(A)!0.5!(L1)$) {$a_{1}$}; \node[right] at
    ($(B)!0.5!(L2)$) {$a_{2}$}; \node[left] at ($(C)!0.5!(L3)$) {$a_{3}$};
    \node[below] at ($(A)!0.5!(B)$) {$a_{4}$}; \node[right] at
    ($(C)!0.5!(B)$) {$a_{5}$}; \node[left] at ($(C)!0.5!(A)$) {$a_{6}$};
    \node at (1.05,-2.1) {$\mathcal{N}_{1}$};
  \end{tikzpicture}
  \hspace{1cm}
  \begin{tikzpicture}[scale=1, transform shape,
    >={Stealth[length=4pt,width=3pt]}, line cap=round, line join=round]
    \coordinate (A) at (0,0); 
    \coordinate (B) at (2.1,0.15); 
    \coordinate (C) at (1.15,1.45); 
    \coordinate (L1) at (-0.35,-1.2); 
    \coordinate (L2) at (2.45,-1.2); 
    \coordinate (L3) at (1.15,2.6); 
    \draw[-, thick] (A) -- (L1); \draw[-, thick] (B) -- (L2); \draw[-, thick]
    (C) -- (L3); \draw[thick] (C) -- (B); \draw[->, shorten >=1pt, shorten
    <=1pt, thick, dashed] (B) -- (A); \draw[->, shorten >=1pt, shorten <=1pt,
    thick, dashed] (C) -- (A);
    \filldraw[black] (A) circle (1pt); \filldraw[black] (B) circle (1pt);
    \filldraw[black] (C) circle (1pt); \filldraw[black] (L1) circle (1pt);
    \filldraw[black] (L2) circle (1pt); \filldraw[black] (L3) circle (1pt);
    \node[below] at (L1) {2}; \node[below] at (L2) {1}; \node[above] at (L3)
    {3};
    \node[left] at ($(A)!0.5!(L1)$) {$b_{2}$}; \node[right] at
    ($(B)!0.5!(L2)$) {$b_{1}$}; \node[left] at ($(C)!0.5!(L3)$) {$b_{3}$};
    \node[below] at ($(A)!0.5!(B)$) {$b_{4}$}; \node[right] at
    ($(C)!0.5!(B)$) {$b_{5}$}; \node[left] at ($(C)!0.5!(A)$) {$b_{6}$};
    \node at (1.05,-2.1) {$\mathcal{N}_{2}$};
  \end{tikzpicture}
  \hspace{1cm}
  \begin{tikzpicture}[scale=1, transform shape,
    >={Stealth[length=4pt,width=3pt]}, line cap=round, line join=round]
    \coordinate (A) at (0,0); 
    \coordinate (B) at (2.1,0.15); 
    \coordinate (C) at (1.15,1.45); 
    \coordinate (L1) at (-0.35,-1.2); 
    \coordinate (L2) at (2.45,-1.2); 
    \coordinate (L3) at (1.15,2.6); 
    \draw[-, thick] (A) -- (L1); \draw[-, thick] (B) -- (L2); \draw[-, thick]
    (C) -- (L3); \draw[thick] (C) -- (B); \draw[->, shorten >=1pt, shorten
    <=1pt, thick, dashed] (B) -- (A); \draw[->, shorten >=1pt, shorten <=1pt,
    thick, dashed] (C) -- (A);
    \filldraw[black] (A) circle (1pt); \filldraw[black] (B) circle (1pt);
    \filldraw[black] (C) circle (1pt); \filldraw[black] (L1) circle (1pt);
    \filldraw[black] (L2) circle (1pt); \filldraw[black] (L3) circle (1pt);
    \node[below] at (L1) {3}; 
    \node[below] at (L2) {1}; 
    \node[above] at (L3) {2};
    \node[left] at ($(A)!0.5!(L1)$) {$c_{3}$}; \node[right] at
    ($(B)!0.5!(L2)$) {$c_{1}$}; \node[left] at ($(C)!0.5!(L3)$) {$c_{2}$};
    \node[below] at ($(A)!0.5!(B)$) {$c_{4}$}; \node[right] at
    ($(C)!0.5!(B)$) {$c_{5}$}; \node[left] at ($(C)!0.5!(A)$) {$c_{6}$};
    \node at (1.05,-2.1) {$\mathcal{N}_{3}$};
  \end{tikzpicture}
  
  \caption{The three semi-directed networks corresponding to models
    $\mathcal{M}_1, \mathcal{M}_2,$ and $\mathcal{M}_3$, with associated
    Fourier edge parameters.}
  \label{fig:models-N1-N2-N3-with-labels}
\end{figure}
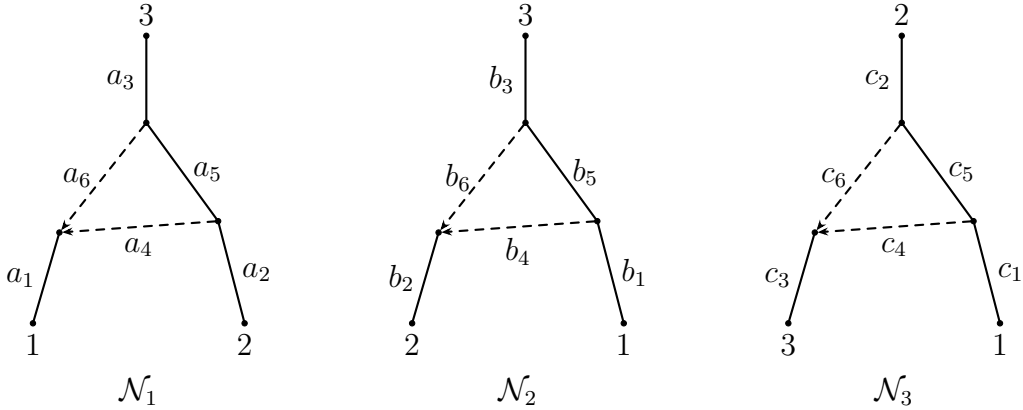

We begin by noting that because the Fourier transform gives a linear change of
coordinates, for the network model, each $q_\omega$ is parameterized by a
convex combination of the parameterizations for the displayed trees in Fourier
coordinates. For example, if we delete the edge labeled $a_6$
with probability $\delta$ and the
$a_4$ edge
with probability $(1 - \delta)$ in the Jukes-Cantor model on the network
$\mathcal{N}_1$ shown in \Cref{fig:models-N1-N2-N3-with-labels}, then
$ q_{\mathtt{CGT}}$, $q_{\mathtt{GGA}}$, and $q_{\mathtt{AAA}}$ are parameterized as follows:
\begin{align*}
  q_{\mathtt{CGT}} &= \delta a^{1}_{\mathtt{C}}a^{2}_{\mathtt{G}}a^{3}_{\mathtt{T}}a^{4}_{\mathtt{C}}a^{5}_{\mathtt{T}} + (1-\delta)a^{1}_{\mathtt{C}}a^{2}_{\mathtt{G}}a^{3}_{\mathtt{T}}a^{5}_{\mathtt{G}}a^6_{\mathtt{C}}    \\
  q_{\mathtt{GGA}} &= \delta a^{1}_{\mathtt{G}}a^{2}_{\mathtt{G}}a^{3}_{\mathtt{A}}a^{4}_{\mathtt{G}}a^{5}_{\mathtt{A}} + (1-\delta)a^{1}_{\mathtt{G}}a^{2}_{\mathtt{G}}a^{3}_{\mathtt{A}}a^{5}_{\mathtt{G}}a^6_{\mathtt{G}}    \\
  q_{\mathtt{AAA}} &= \delta a^{1}_{\mathtt{A}}a^{2}_{\mathtt{A}}a^{3}_{\mathtt{A}}a^{4}_{\mathtt{A}}a^{5}_{\mathtt{A}} + (1-\delta)a^{1}_{\mathtt{A}}a^{2}_{\mathtt{A}}a^{3}_{\mathtt{A}}a^{5}_{\mathtt{A}}a^6_{\mathtt{A}}.
\end{align*}

From \Cref{eq:fourier-parameterization}, we see that many coordinates for a
group-based model on a tree are trivial after transformation, and the same
coordinates will be trivial for the network model. Thus, for the 3-leaf
network model, there are only 16 non-trivial coordinates to consider. We can
also simplify the parameterization of many of these non-trivial coordinates.
Since $a^i_{\mathtt{A}}= 1$ for all $i\in[6]$, it is not necessary to include the $a^i_{\mathtt{A}}$
parameters (and hence, no need to consider the coordinate $q_{\mathtt{AAA}}$
which is equal to 1 for any choice of parameters for the model). Moreover,
since $a^i_{\mathtt{C}} = a^i_{\mathtt{G}}= a^i_{\mathtt{T}}$, several of the Fourier coordinates
are always identical; for example,
$q_{\mathtt{CCA}} = q_{\mathtt{GGA}} = q_{\mathtt{TTA}}$.

After making these simplifications, there are only four non-trivial
equivalence classes of Jukes-Cantor $q$-coordinates. We choose one
representative from each of these equivalence classes,
$q_{\mathtt{ACC}},q_{\mathtt{CAC}}, q_{\mathtt{CCA}},$ and $q_{\mathtt{CGT}}$,
to form the set of \emph{simplified Jukes-Cantor $q$-coordinates}. For
$\mathcal{N}_1$, the simplified Jukes-Cantor $q$-coordinates are parameterized
as follows:

\begin{equation}
  \label{eq:q-defs}
  \begin{aligned}
    q_{\mathtt{ACC}} &= a_2a_3a_5  \\
    q_{\mathtt{CAC}} &= \delta a_1a_3a_4a_5 \hphantom{lll} + (1-\delta)a_1a_3a_6 \\
    q_{\mathtt{CCA}} &= \delta a_1a_2a_4 \hphantom{llllll} + (1-\delta)a_1a_2a_5a_6 \\
    q_{\mathtt{CGT}} &= \delta a_1a_2a_3a_4a_5 + (1-\delta)a_1a_2a_3a_5a_6
  \end{aligned}
\end{equation}

\noindent where $a_{1},\ldots,a_{6}\in (0,1)$ and $\delta\in (0,1)$.

Working in the simplified Jukes-Cantor $q$-coordinates, the parameter space
becomes
\begin{equation*}
  \Theta_{\mathcal{N}_1} = \left\{(a_{1},\ldots,a_{6}):
    0<a_{1},\ldots,a_{6}<1\right\},
\end{equation*}
the parameterization map becomes
\begin{align*}
  \psi_{\mathcal{N}_{1}}:\mathbb{R}^7 & \to \mathbb{R}^{4} \\
  (a_1,a_2,a_3,a_4,a_5,a_6,\delta)& \mapsto (q_{\mathtt{ACC}},q_{\mathtt{CAC}},q_{\mathtt{CCA}},q_{\mathtt{CGT}}),
\end{align*}
and the model becomes
\begin{equation*}
  \mathcal{M}_{1} = \left\{\psi_{\mathcal{N}_{1}}(\theta): \theta\in
    \Theta_{\mathcal{N}_{1}}\times (0,1)\right\} \subseteq \mathbb{R} ^4.
\end{equation*}

\section{Results}
\label{sec:results}

In this section, we begin by describing the non-identifiability of the numerical parameters for a single 3-leaf triangle network model.

\subsection{Identifiability of Numerical Parameters}
\label{sec:identifiability-of-numerical-parameters}

The numerical parameters of the 3-leaf network model are not identifiable. Indeed, only the products $\delta a_1a_4$ and $(1-\delta)a_1a_6$ are recoverable at best. Thus, for 
any choice of $a_1, a_2,a_3,a_4,a_5, a_6, \delta \in(0,1)$, we can describe a two-dimensional region $R$  
of parameter space for which the map $\psi_{\mathcal{N}_1}$  sends every
point of $R$ to the same point in the model. For example, for any $\theta=(a_{1},\ldots,a_{6},\delta)\in \Theta_{\mathcal{N}_{1}}\times (0,1)$, we
can construct a new set of parameters $\theta'\in \Theta_{\mathcal{N}_{1}}\times (0,1)$
such that $\psi_{\mathcal{N}_{1}}(\theta)=\psi_{\mathcal{N}_1}(\theta')$ by taking
\begin{equation}
  \label{eq:-scale-first-NEW}
  \theta' 
  = \left(\beta_1a_1,a_{2},a_{3},\frac{a_4}{\beta_1 \beta_2},a_{5},
    \frac{a_6(1-\delta)}{\beta_1(1-\beta_2 \delta)},\beta_2\delta\right)
\end{equation}
for any choice of $\beta_1,\beta_2$ in appropriate open intervals of the positive real numbers containing~1.

\begin{figure*}
  \centering
  \includegraphics[scale=.26]{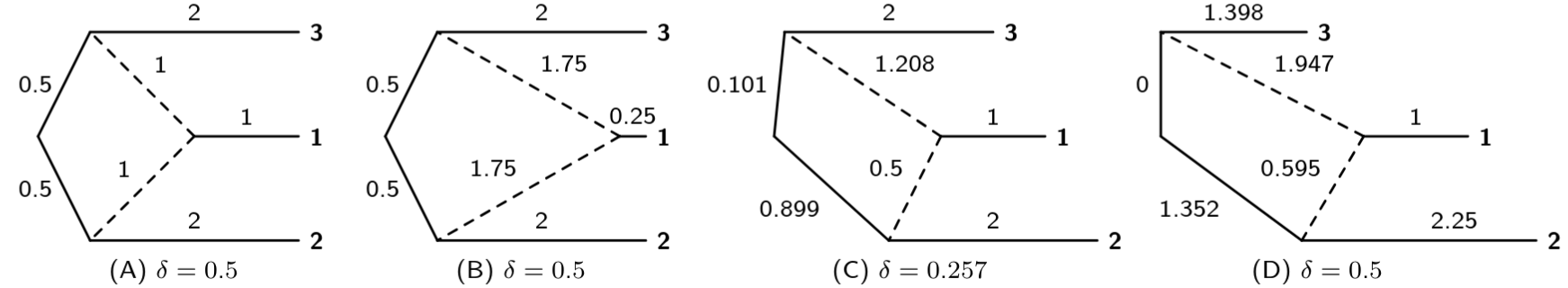}
  \caption{Networks (A)-(D) show four different choices of branch lengths
    for $\mathcal{N}_1$ that all give rise to the same site-pattern
    probability distribution. The networks are rooted along edge $e_5$ for
    display, but as noted above, the location of the root is not identifiable.
    Branch lengths, measured in expected number of mutations per site, are
    displayed in the figure as horizontal distances.}
  \label{fig:unidentifiable-networks}
\end{figure*}

The networks shown in \Cref{fig:unidentifiable-networks} illustrate this
non-identifiability of the numerical parameters. First, we assign branch
lengths to the network $\mathcal{N}_{1}$ in \Cref{fig:models-N1-N2-N3-with-labels} to obtain
network~(A). For network (A), the edges $e_1, e_4, e_5$ and $e_6$ are all length 1, the
edges $e_2$ and $e_3$ are length 2, and $\delta = \frac{1}{2}.$ Recall that
the Fourier parameter corresponding to the edge $e_i$ of length $t_i$ is
$a_i = e^{-4t_i/3}$. Thus, the Fourier parameters for network (A) are
$a_1 = a_4 = a_5 = a_6 = e^{-4/3}$ and $a_2 = a_3 = e^{-8/3}$. Networks (B)
and (C) are obtained by scaling the parameters according to
\cref{eq:-scale-first-NEW} and then converting the Fourier parameters to
branch lengths. In particular, network (B) is obtained by setting
$\beta_1 = e$ and $\beta_2 = 1$, and network (C) by setting $\beta_1 = 1$ and
$\beta_2 = e^{-2/3}$.

To account for the non-identifiability of $a_{1},a_{4},a_{6},$ and $\delta$, in
our proof of \Cref{thm:m1-semialgebraic-description-in-qs} we will
reparameterize the model by making the substitutions
$a_{4}^{*} = \delta a_{1}a_{4}$ and $a_{6}^{*}=(1-\delta)a_{1}a_{6}$ (see
\Cref{app:proof-of-main-thm} for details). The new parameters $a_{4}^{*}$ and
$a_{6}^{*}$ must satisfy the additional restriction that
$a_{4}^{*}+a_{6}^{*}<1$ but are otherwise free. Yet, even after introducing
$a_4^*$ and $a_6^*$, the map is still not one-to-one. For a generic point
$\textbf{q}\in\mathcal M_1$, the fiber is a one-dimensional region of
$\mathbb R^5$, which is described in detail in
\cref{lem:existence-of-a-family-of-solutions}. Thus, since there is an
additional degree-of-freedom, if we perturb one of the Fourier parameters, we
can adjust the others so that the corresponding point in the model remains
unchanged. Network (D) in \Cref{fig:unidentifiable-networks} is one such
network, where beginning with network (A), we increase the length of edge
$e_2$ by $1/4$, and adjust all other parameters accordingly.

\subsection{A Semialgebraic Description of the Jukes-Cantor 3-Leaf Triangle Network Models}
\label{sec:a-semialgebraic-description}

We now give a complete semialgebraic description of $\mathcal{M}_1$ in the
space of simplified Jukes-Cantor $q$-coordinates. Combined with the results
above, this will allow us to characterize the intersections of $\mathcal M_1$,
$\mathcal M_2$, and $\mathcal M_3$ to gain insight into the question of
network identifiability.

\begin{theorem}[Semialgebraic description of $\mathcal{M}_{1}$]
  \label{thm:m1-semialgebraic-description-in-qs}
  Let $\mathbf{q}=(q_{\mathtt{ACC}},q_{\mathtt{CAC}},q_{\mathtt{CCA}},q_{\mathtt{CGT}})$. Then
  $\mathbf{q}\in \mathcal{M}_{1}$ if and only if the following inequalities
  hold:
  \begin{align}
    0<q_{\mathtt{ACC}},q_{\mathtt{CAC}},q_{\mathtt{CCA}},q_{\mathtt{CGT}}&<1 \label{ineq:0<q<1}\\
    q_{\mathtt{ACC}} -   q_{\mathtt{CGT}} & >0 \label{ineq:ACC-CGT}\\ 
    q_{\mathtt{CGT}} - q_{\mathtt{ACC}}q_{\mathtt{CAC}} &> 0 \label{ineq:CGT-ACC.CAC} \\
    q_{\mathtt{ACC}}q_{\mathtt{CAC}}q_{\mathtt{CCA}}-q_{\mathtt{CGT}}^{2}&>0 \label{ineq:hybrid}\\
    q_{\mathtt{CGT}}+q_{\mathtt{ACC}}(q_{\mathtt{CGT}}-q_{\mathtt{CAC}}-q_{\mathtt{CCA}}) &>0. \label{ineq:a23}
  \end{align}
  Moreover, if $\textbf{q}\in \mathcal{M}_{1}$ then the following inequalities
  also hold:
  \begin{align}
    q_{\mathtt{CAC}} -   q_{\mathtt{CGT}} & >0  \label{ineq:CAC-CGT}\\ 
    q_{\mathtt{CCA}} -   q_{\mathtt{CGT}} & >0  \label{ineq:CCA-CGT}\\ 
    q_{\mathtt{CGT}} - q_{\mathtt{ACC}}q_{\mathtt{CCA}} &> 0.  \label{ineq:CGT-ACC.CCA}
  \end{align}
\end{theorem}
  
We defer the proof of \Cref{thm:m1-semialgebraic-description-in-qs} to
\Cref{app:proof-of-main-thm}. Note that by symmetries, we can easily obtain
similar descriptions of $\mathcal{M}_2$ and $\mathcal{M}_3$ by permuting the
indices appropriately, i.e., by swapping the first and second index or by
swapping the first and third. Note that inequalities \eqref{ineq:0<q<1},
\eqref{ineq:ACC-CGT},\eqref{ineq:hybrid}, \eqref{ineq:CAC-CGT}, and
\eqref{ineq:CCA-CGT} appear in the description for all three models, and thus,
are not informative about the location of the hybrid node. Inequalities
\eqref{ineq:CGT-ACC.CAC} and \eqref{ineq:CGT-ACC.CCA} each appear in two of
the three model descriptions and thus are only partially informative.

The linear inequalities \eqref{ineq:ACC-CGT}, \eqref{ineq:CAC-CGT}, and
\eqref{ineq:CCA-CGT} are readily interpretable in terms of the two displayed
trees. For example, if we restrict to the parameterization for the tree
containing edge $e_4$ (this can be done by setting $\delta=1$), we see that
inequality \eqref{ineq:ACC-CGT} says that leaf node $1$ in
\Cref{fig:models-N1-N2-N3-with-labels} does not lie on the path between leaves
$2$ and $3$ in the displayed tree containing $e_4$. To see why this is the
case, observe that setting $\delta = 1$ in \cref{eq:q-defs} yields
\begin{align}
  \label{eq:interpret-first-three-inequalities}
  q_{\mathtt{ACC}}-q_{\mathtt{CGT}} 
  &=a_{2}a_{3}a_{5}(1-   a_{1}a_{4} )\geq 0.
\end{align}
If the edges $e_2, e_3$, and $e_5$ are finite then this is 0 if and only if
$a_{1}=a_{4}=1$ or equivalently, $t_{1}=t_{4}=0$. A similar situation holds
when considering the displayed tree containing $e_6$. Since the network
distribution is a convex sum of the tree distributions, the linear inequality
\cref{eq:interpret-first-three-inequalities} holds strictly for the network as
well.

Inequality \eqref{ineq:hybrid} also has an interesting interpretation that was
first noted by \cite{englander2025identifiability}. Substituting a
site-pattern distribution from any of the three 3-leaf networks into this
polynomial will yield a positive value. But, the polynomial
$q_{\mathtt{ACC}}q_{\mathtt{CAC}}q_{\mathtt{CCA}}-q_{\mathtt{CGT}}^{2}$ is an algebraic invariant for the 3-leaf
claw tree under the Jukes-Cantor model. Thus, this invariant distinguishes the
3-leaf tree model from the 3-leaf network models, and the invariant residual
may provide evidence of hybridization.

The other inequalities in \cref{thm:m1-semialgebraic-description-in-qs} are
less illuminating. When converted to the probability coordinates, inequalities
\eqref{ineq:CGT-ACC.CAC} and \eqref{ineq:CGT-ACC.CCA} are irreducible degree 2
polynomials with 20 terms and inequality \eqref{ineq:a23} is an irreducible
degree 2 polynomial with 18 terms, and none of these have a readily apparent
biological interpretation.

\subsection{Model Intersections and Implications for Identifiability}
\label{sec:model-intersections-and-implications}

In this section, we prove our main result (\Cref{thm:not-identifiable})
regarding the identifiability of 3-cycles in a phylogenetic network model. In
practical terms, \Cref{thm:not-identifiable} implies that if edge parameters
are chosen randomly, then with positive probability, determining which of the
three leaves is descended from the hybrid node is impossible even with
infinite data.

First, we show the sets $\mathcal{M}_{1}\backslash\mathcal{M}_{2}$ and
$\mathcal{M}_{1}\cap\mathcal{M}_{2}$ are separated by a single algebraic
curve; namely, $q_{\mathtt{CGT}}+q_{\mathtt{CAC}}(q_{\mathtt{CGT}}-q_{\mathtt{ACC}}-q_{\mathtt{CCA}})=0$.

\begin{theorem}
  \label{thm:q-coordinate-conditions-for-belonging-to-m1-intersect-m2}
  Let $\mathbf{q}=(q_{\mathtt{ACC}},q_{\mathtt{CAC}},q_{\mathtt{CCA}},q_{\mathtt{CGT}})\in \mathcal{M}_{1}$. Then
  $\mathbf{q}\in \mathcal{M}_{2}$ if and only if
  \begin{equation}\label{eq:M1-M2-condition}
    q_{\mathtt{CGT}}+q_{\mathtt{CAC}}(q_{\mathtt{CGT}}-q_{\mathtt{ACC}}-q_{\mathtt{CCA}})>0.
  \end{equation}
  Similarly, $\mathbf{q}\in \mathcal{M}_{3}$ if and only if
  \begin{equation}\label{eq:M1-M3-condition}
    q_{\mathtt{CGT}}+q_{\mathtt{CCA}}(q_{\mathtt{CGT}}-q_{\mathtt{ACC}}-q_{\mathtt{CAC}})>0.
  \end{equation}
\end{theorem}
\begin{proof}
  We prove only the case for $\mathcal{M}_{2}$ as the case for
  $\mathcal{M}_{3}$ is similar. As noted in the discussion following
  \Cref{thm:m1-semialgebraic-description-in-qs}, by permuting indices of the
  Jukes-Cantor $q$-coordinates, we can easily derive defining sets of
  inequalities for the other two network models. In particular, by transposing
  the first two indices in the inequalities \eqref{ineq:0<q<1} through
  \eqref{ineq:a23}, and using the linear invariant $q_{\mathtt{CGT}}-q_{\mathtt{GCT}}=0$ to
  write $q_{\mathtt{GCT}}$ as $q_{\mathtt{CGT}}$, we find that the set $\mathcal{M}_{2}$ is
  defined by the inequality \eqref{eq:M1-M2-condition} together with the
  inequalities
  \begin{equation}\label{eq:M2-inequalities}
    \begin{aligned}
      0<q_{\mathtt{ACC}},q_{\mathtt{CAC}},q_{\mathtt{CCA}},q_{\mathtt{CGT}}&<1\\
      q_{\mathtt{CAC}} -   q_{\mathtt{CGT}} & >0 \\
      q_{\mathtt{CGT}} - q_{\mathtt{CAC}}q_{\mathtt{ACC}} &> 0 \\
      q_{\mathtt{ACC}}q_{\mathtt{CAC}}q_{\mathtt{CCA}}-q_{\mathtt{CGT}}^{2}&>0.
    \end{aligned}
  \end{equation}
  Suppose $\mathbf{q}\in \mathcal{M}_{1}$. It is easy to see that all of the
  inequalities in \eqref{eq:M2-inequalities} are implied by inequalities in
  the statement of \Cref{thm:m1-semialgebraic-description-in-qs}, since
  $\mathbf{q}\in \mathcal{M}_{1}$. Hence, $\mathbf{q}\in \mathcal{M}_{2}$ if
  and only if inequality \eqref{eq:M1-M2-condition} holds.
\end{proof}

\Cref{thm:q-coordinate-conditions-for-belonging-to-m1-intersect-m2}
establishes a criterion for determining whether a point
$\mathbf{q}\in \mathcal{M}_{1}$ is in
$\mathcal{M}_{1}\backslash\mathcal{M}_{2}$ or
$\mathcal{M}_{1}\cap \mathcal{M}_{2}$. The next corollary, which will be
useful for interpretation, expresses this inequality in terms of the numerical
parameters $a_1,\ldots,a_6,$ and $\delta$, and follows immediately by
substituting the parameterization of $\mathcal{M}_{1}$
into the inequalities \eqref{eq:M1-M2-condition} and
\eqref{eq:M1-M3-condition}.

\begin{corollary}[Distinguishability criteria]
  \label{thm:parameter-conditions-for-belonging-to-m1-intersect-m2}
  Let $\mathbf{q}\in \mathcal{M}_{1}$, and let
  $\theta=(a_{1},a_{2},a_{3},a_{4},a_{5},a_{6},\delta)\in \Theta_{\mathcal{N}_{1}}\times (0,1)$ be a
  choice of numerical parameters such that $\psi_{\mathcal{N}_{1}}(\theta) = \mathbf{q}$. Then
  $\mathbf{q}\notin \mathcal{M}_{2}$ if and only if

  \begin{equation}\label{eq:solutionCondition-M1-M2}
    \begin{aligned} 
      a_{1} &\geq \frac{a_{5}}{ \delta a_{4}a_{5}+ a_{6}(1-\delta)} 
            \cdot \frac{ \delta a_{4}(1-a_{3}a_{5})+  a_{6}(1-\delta)(1-a_{3})}{ \delta a_{4}(1-a_{3}a_{5})+ a_{5} a_{6}(1-\delta)(1-a_{3})}.
    \end{aligned}
  \end{equation} 
  Similarly, $\mathbf{q}\notin \mathcal{M}_{3}$ if and only if
  \begin{equation}\label{eq:solutionCondition-M1-M3}
    \begin{aligned}
      a_{1} &\geq \frac{a_{5}}{ \delta {a}_{4}+ a_{5}a_{6}(1-\delta)}  
      \cdot \frac{ a_{6}(1-\delta)(1-a_{2}a_{5})+  \delta {a}_{4}(1-a_{2})}{ a_{6}(1-\delta)(1-a_{2}a_{5})+  \delta {a}_{4}a_{5}(1-a_{2})}.
    \end{aligned}
  \end{equation}
\end{corollary}

The network parameter of these models is \emph{generically identifiable} if
the set of parameters that maps into the intersection of two models is measure
zero within the parameter space. \Cref{thm:not-identifiable} states that this
is not the case. To prove this, we will utilize the following lemma, in which
we use the notation $\overline{\mathcal{M}}_{2}$ to denote the Euclidean
closure of $\mathcal{M}_{2}$.

\begin{lemma}
  \label{lem:nonempty-open-subsets}
  Both $\mathcal{M}_{1}\cap \mathcal{M}_{2}$ and
  $\mathcal{M}_{1}\backslash \overline{\mathcal{M}}_{2}$ are nonempty open
  subsets of $\mathbb{R}^{4}$.
\end{lemma}
\begin{proof}
  We begin by showing that $\mathcal{M}_{1}\cap \mathcal{M}_{2}$ and
  $\mathcal{M}_{1}\backslash \overline{\mathcal{M}}_{2}$ are open sets. By
  \Cref{thm:m1-semialgebraic-description-in-qs}, $\mathcal{M}_{1}$ is the
  intersection of a finite number of open sets; namely, the sets
  corresponding to the strict inequalities in
  \eqref{ineq:0<q<1}-\eqref{ineq:a23}. Therefore $\mathcal{M}_{1}$ is open.
  Similarly,
  \cref{thm:q-coordinate-conditions-for-belonging-to-m1-intersect-m2} implies
  that $\mathcal{M}_{1}\cap \mathcal{M}_{2}$ is open since it is the
  intersection of two open sets. Moreover, since $\mathcal{M}_{1}$ is
  open and $\overline{\mathcal{M}}_{2}$ is closed,
  $\mathcal{M}_{1}\backslash \overline{\mathcal{M}}_{2}$ is open as well.

  It remains to show that $\mathcal{M}_{1}\cap \mathcal{M}_{2}$ and
  $\mathcal{M}_{1}\backslash \overline{\mathcal{M}}_{2}$ are nonempty. To show
  that $\mathcal{M}_{1}\cap \mathcal{M}_{2}\neq \emptyset$, observe that for
  any choice of $a_{1},a_{2},a_{3},a_{4},a_{6},\delta\in (0,1)$, as $a_{5}\to 1$, the
  right-hand side of inequality \eqref{eq:solutionCondition-M1-M2} tends to
  \begin{equation*}
    \frac{1}{\delta a_{4}+(1-\delta)a_{6}}>1.
  \end{equation*}
  Therefore we may choose $a_{5}$ sufficiently close to $1$ so that the
  right-hand side of inequality \eqref{eq:solutionCondition-M1-M2} is larger
  than $a_1$, in which case
  \cref{thm:parameter-conditions-for-belonging-to-m1-intersect-m2} implies
  $\psi_{\mathcal{N}_{1}}(a_{1},a_{2},a_{3},a_{4},a_{5},a_{6},\delta)\in
  \mathcal{M}_{1}\cap \mathcal{M}_{2}$, and hence
  $\mathcal{M}_{1}\cap \mathcal{M}_{2}\neq \emptyset$. Similarly, to show that
  $\mathcal{M}_{1}\backslash \overline{\mathcal{M}}_{2}\neq \emptyset$,
  observe that we may choose $a_{5}$ sufficiently close to zero so that the
  right-hand side of \eqref{eq:solutionCondition-M1-M2} is smaller
  than $a_{1}$, in which case
  $\mathbf{q}=\psi_{\mathcal{N}_{1}}(a_{1},a_{2},a_{3},a_{4},a_{5},a_{6},\delta)\in
  \mathcal{M}_{1}\backslash \overline{\mathcal{M}}_{2}$.
\end{proof}

\cref{lem:nonempty-open-subsets} can be understood as saying that
$\mathcal{M}_{1}\cap \mathcal{M}_{2}$ and
$\mathcal{M}_{1}\backslash \mathcal{M}_{2}$ are both full-dimensional subsets
of $\mathbb{R}^{4}$. The proof of \cref{thm:not-identifiable} now follows by a
simple topological argument. In particular, since
$\mathcal{M}_{1}\cap \mathcal{M}_{2}$ is open by
\cref{lem:nonempty-open-subsets}, it follows by continuity of
$\psi_{\mathcal{N}_{1}}$ that
$\psi_{\mathcal{N}_{1}}^{-1}(\mathcal{M}_{1}\cap \mathcal{M}_{2})$ is an open
subset of the parameter space
$\Theta_{\mathcal{N}_{1}}\times (0,1)\subseteq \mathbb{R}^{7}$. Therefore,
since $\psi_{\mathcal{N}_{1}}^{-1}(\mathcal{M}_{1}\cap \mathcal{M}_{2})$ is
open, it is also both full-dimensional and of positive measure.

\begin{theorem}
  \label{thm:not-identifiable}
  The semi-directed topology of a 3-leaf Jukes-Cantor triangle network model
  is \textbf{not} generically identifiable from the site pattern distribution.
\end{theorem}

Although we have now shown that the network parameter is not generically
identifiable, the story does not end here. Since
$\mathcal{M}_{1}\backslash \overline{\mathcal{M}}_{2}$ is open, it follows
that
$\psi_{\mathcal{N}_{1}}^{-1}(\mathcal{M}_{1}\backslash
\overline{\mathcal{M}}_{2})$ is open, and hence
$\psi_{\mathcal{N}_{1}}^{-1}(\mathcal{M}_{1}\backslash \mathcal{M}_{2})$
contains an open set. Thus, there exists a full dimensional subset of
parameter space on which it is theoretically possible to determine which of
the three leaves is descendant from the hybrid node.

\subsection{The Volume of Model Intersections} 
\label{sec:volume-calculations}

In this section, we take a closer look at the models and their
intersections to better understand the practical implications of these results
for phylogenetic inference. To estimate the relative volumes of the three
models and their intersections, we sampled $n=10^{9}$ points uniformly at
random from the probability simplex $\Delta_{4}$ using the
$\text{Dirichlet}(1,1,1,1,1)$ distribution and recorded the proportion of
points whose Fourier transform satisfied the model inequalities. These results
are summarized in \Cref{fig:venn-simplex}. The code for this and all
subsequent simulations in this section are available in the file
\texttt{model-size-simulations.jl} in the Supplementary Materials.

\begin{figure}[b]
  \centering
  \includegraphics[scale=.25]{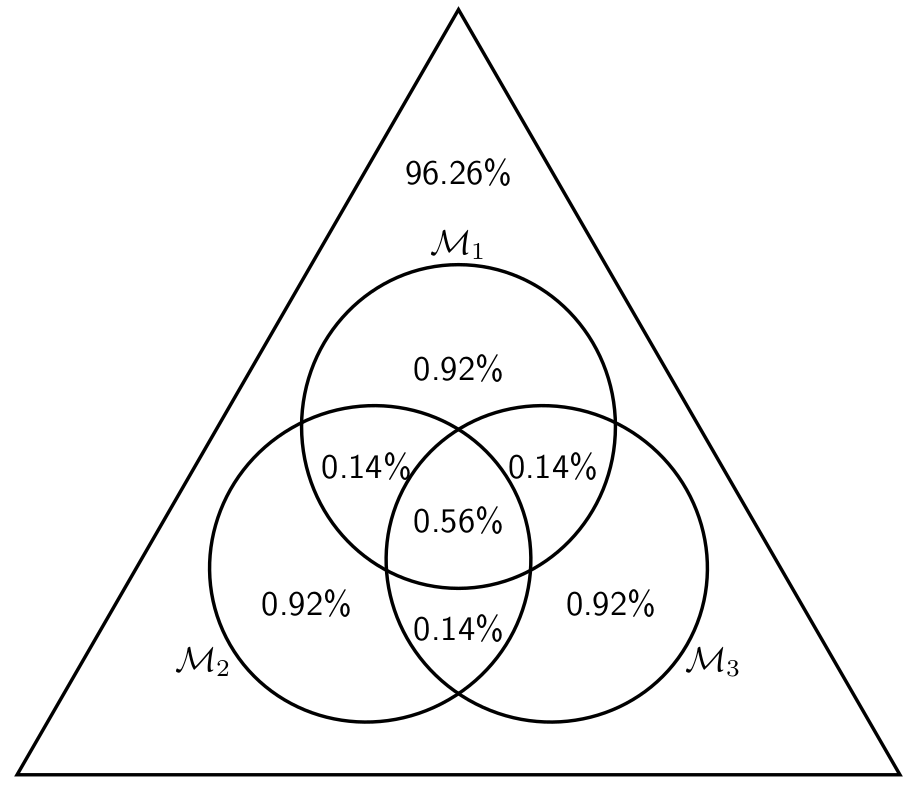}
  \caption{Venn diagram showing the percentage of points in $\Delta_{4}$ which
    belong to the regions of intersections of the models
    $\mathcal{M}_{1},\mathcal{M}_{2}$, and $\mathcal{M}_{3}$. The vast
    majority of simplex points $(96.2\%)$ do not correspond to any of the
    three network models.}
  \label{fig:venn-simplex}
\end{figure}

\begin{observation}[Model overlap in the simplex]
  \label{obs:intersection-volumes-in-qs}
  The three models $\mathcal{M}_{1},\mathcal{M}_{2}$ and
  $\mathcal{M}_{3}$ correspond to small regions of the simplex which overlap substantially. Only about $3.7\%$ of simplex points
  correspond to distributions in any of the three models, with each model accounting for only about $1.76\%$ of
  simplex points. Moreover, of the simplex points which belong to a given
  model $\mathcal{M}_{i}$, almost half $(47.9\%)$ coincide with at least one
  of the other two models.
  
  Since the Fourier transformation scales 4-dimensional volume in the simplex
  uniformly by a constant factor (see \cref{app:fourier-transform}), we can
  estimate the model volumes explicitly. The 4-dimensional volume of the space
  of simplified Jukes-Cantor coordinates is approximately $0.26337$; within this space, each of
  the models $\mathcal{M}_{1},\mathcal{M}_{2},$ and $\mathcal{M}_{3}$ has absolute volume
  approximately $0.00464$, the volume of the intersection of any two models is
  approximately $0.00184$, and the volume of the intersection of all three is
  approximately $0.00147$. See \cref{rmk:volume-scaling} in \cref{app:fourier-transform} for details.
\end{observation}

While these figures give some perspective on the geometry of the models and
their intersections in the $q$-coordinates, the issue of identifiability is
perhaps better illuminated by considering which \textit{parameter} choices for a model
give a site-pattern probability distribution that belongs to one or both of
the other models. To investigate this, we sampled  $n=10^{8}$ numerical parameter vectors $(a_{1},\ldots,a_{6},\delta)$
for $\mathcal{N}_{1}$ uniformly at random from 
$\Theta_{\mathcal{N}_{1}}\times (0,1)$.

\begin{observation}[Model overlap in the parameter space]
  \label{obs:volume-of-parameter-intersections}
  We find that when parameters are sampled uniformly at random, about
  $94.3\%$ of parameter choices yield a point in the intersection of two or
  more models, in which case the hybrid is not distinguishable. In particular,
  about $91.4\%$ of parameter choices lie in
  $\psi_{\mathcal{N}_{1}}^{-1}(\mathcal{M}_{1}\cap \mathcal{M}_{2}\cap
  \mathcal{M}_{3})$, about $92.8\%$ lie in each of
  $\psi_{\mathcal{N}_{1}}^{-1}(\mathcal{M}_{1}\cap \mathcal{M}_{2})$ and
  $\psi_{\mathcal{N}_{1}}^{-1}(\mathcal{M}_{1}\cap \mathcal{M}_{3})$, and only
  $5.7\%$ lie in
  $\psi_{\mathcal{N}_{1}}^{-1}(\mathcal{M}_{1}\backslash(\mathcal{M}_{2}\cup
  \mathcal{M}_{3}))$.

  One limitation in interpreting these results is that the uniform sampling regime, while of interest geometrically, places excessive
  weight on very long branches which are not biologically realistic.  
  To address this, we simulated networks with shorter, bounded branch lengths. For each
  $\delta\in \left\{.01,.02,\ldots,.99\right\}$ and each
  $m\in \left\{.1, .25, .5, 1\right\}$, we sampled $n=10^{7}$ sets of branch lengths
  $t_{1},\ldots,t_{6}\overset{iid}{\sim} \text{unif}(0,m)$ for $\mathcal{N}_{1}$. 

  \begin{figure}
    \centering
    \begin{tikzpicture}[baseline=(current bounding box.center)]
      \node[inner sep=0] (image) {\includegraphics[scale=.3]{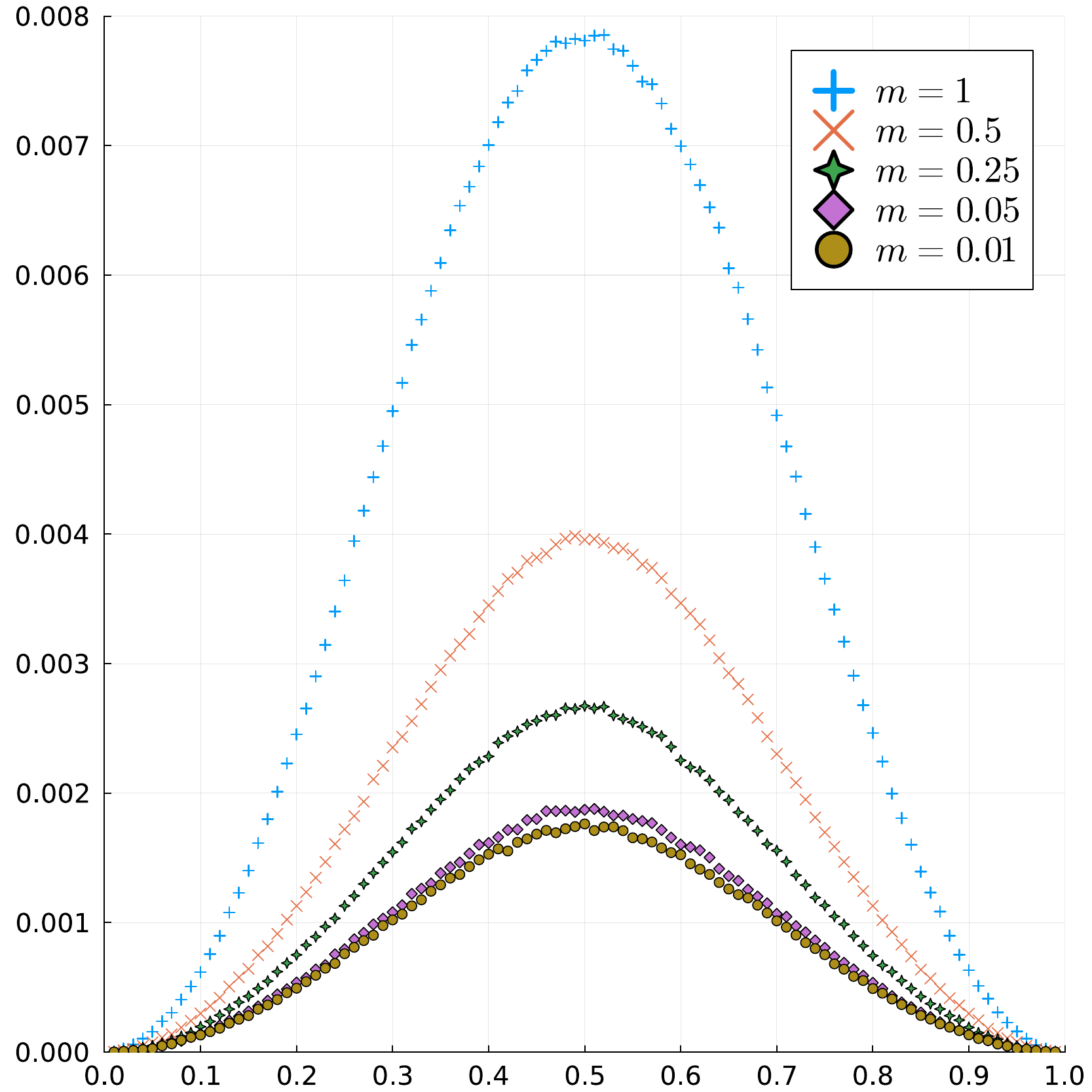}};

      \node[font=\sffamily] at ($(image.south)+(0,-11pt)$) {{$\delta$ (reticulation parameter)}};

      \node[rotate=90,font=\sffamily] at ($(image.west)+(-11pt,0)$)
      {{Proportion of samples in $\mathcal{M}_{1}\backslash \left(\mathcal{M}_{2}\cup \mathcal{M}_{3}\right)$}};
    \end{tikzpicture}
    \caption{The proportion of networks in which the hybrid node is
      distinguishable as a function of $\delta$, for $\mathcal{N}_{1}$ with
      branch lengths
      $t_{1},\ldots,t_{6}\overset{iid}{\sim} \text{unif}(0,m)$.}
    \label{fig:observation-2-plot}
  \end{figure}

  The results of this simulation (shown in \cref{fig:observation-2-plot})
  indicate that the conditions needed for hybrid
  distinguishability are rarely satisfied for the sorts of branch lengths
  typically seen in phylogenetic inference. The proportion of hybrid distinguishable
  networks was always less than $1\%$, even in the best-case setting
  where $\delta=1/2$ and $m=1$. Moreover, as the bound on branch lengths $m$
  decreases, the distinguishability proportion further decreases, and is
  negligible when $\delta$ is close to zero or one.
\end{observation}

\subsection{Implications for Network Inference in Practice} \label{sec:implications}

The previous section shows that the three models overlap substantially,
especially when considering biologically relevant parameters. Here we
elaborate on some of the implications of the distinguishability criteria in
\cref{thm:parameter-conditions-for-belonging-to-m1-intersect-m2} when taking
into account additional biological considerations. We first start with an
example to illustrate
\cref{thm:parameter-conditions-for-belonging-to-m1-intersect-m2}.

\begin{example}

  \label{ex:LargeNetworks-triplet}
  Let $\mathcal N$ be the 5-leaf network shown in \Cref{fig:LargeNetworks} and
  let $X$ be a subset of the taxa $\{A, B, C, D, E\}$. Let $\mathcal{N}_{|X}$
  be the restriction of $\mathcal N$ to the taxa in $X$, $p_{|X}$ the
  site-pattern probability distribution from the Jukes-Cantor model on the
  restricted network when $\delta = \frac{1}{2}$, and $q_{|X}$ the image of
  $p_{|X}$ under the Fourier transform. Notice that when $|X| = 3$,
  $\mathcal{N}_{|X}$ is a tree unless $X$ is equal to $\{A, B, C\}$,
  $\{A, B, E\}$, $\{A, C, D\}$, or $\{A, D, E\}$, in which case it is a 3-leaf
  triangle network.

  As an example, consider the restriction of this network to the leaf set
  $\{A, C, D\}$. Matching the labels from
  \Cref{fig:models-N1-N2-N3-with-labels} to the version of the restricted
  network shown in \Cref{fig:RestrictedLargeNetworks} (where $A,C,$ and $D$
  match to $1,2$, and $3$ respectively), we obtain the following values for
  the Fourier parameters:
  \begin{align*}
    a_{1} &= e^{-\frac{4}{3}(.20)},
    &\quad a_{2} &= e^{-\frac{4}{3}(.30)},
    &\quad a_{3} &= e^{-\frac{4}{3}(.15)},\\
    a_{4} &= e^{-\frac{4}{3}(.10)},
    &\quad a_{5} &= e^{-\frac{4}{3}(.70)},
    &\quad a_{6} &= e^{-\frac{4}{3}(.10)}.
  \end{align*}
  Substituting these values into the right-hand side of inequality
  \eqref{eq:solutionCondition-M1-M2} gives approximately $0.74$, which is less
  than $a_{1}\approx 0.77$. This tells us that $q_{|\{A, C, D\}}$ could not
  have come from the 3-leaf triangle network where $C$ is the leaf below the
  reticulation vertex. On the other hand, applying inequality
  \eqref{eq:solutionCondition-M1-M3} in a similar manner, we find the right
  hand side to be approximately $0.79$, so the inequality is not satisfied.
  This tells us that $q_{|\{A, C, D\}}$ \emph{could have} come from a 3-leaf
  triangle network where $D$ is the leaf below the reticulation vertex. As a
  result, it is \emph{not} possible to determine the orientation of edges in
  the triangle of the 3-leaf network from $q_{|\{A, C, D\}}$.

  Applying the bounds for all of the triplets that result in 3-leaf triangle
  networks, we see that only when $X = \{A, D, E\}$ can we infer the
  orientation of edges in the triangle $\mathcal N_{|X}$ from $q_{|X}$. That
  we can identify the hybrid node in this example only for the two shortest leaf
  edges is in line with our observation from Example \ref{ex:maxexample} that
  the signal of hybridization decays. Note that there is no way to gain
  distinguishability by adding additional leaves below the hybrid node, in
  contrast to the gene tree data setting \citep{allman2024identifiability}.

  \begin{figure}[h]
    \centering
    \includegraphics[height=2in]{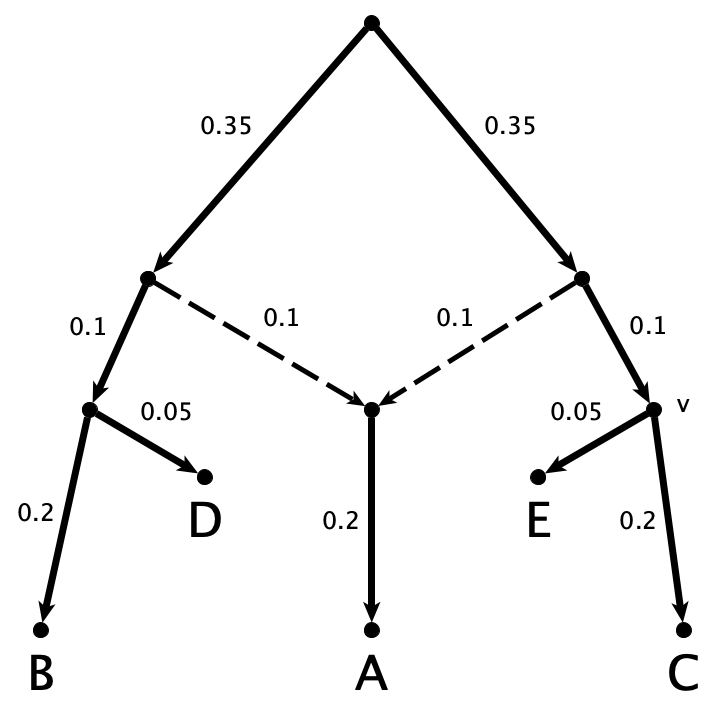}
    \caption{A 5-leaf network with branch lengths in expected number of mutations per site.}
    \label{fig:LargeNetworks}
  \end{figure}

  \begin{figure}[h!]
    \centering
    \includegraphics[width = 4cm]{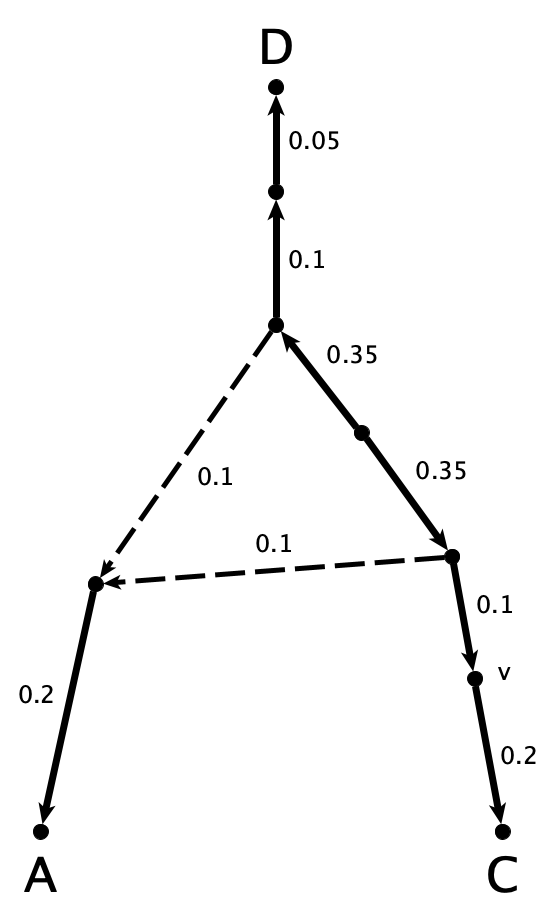}
    \caption{The network in \Cref{fig:LargeNetworks} restricted to $\{A, C, D\}$.}
    \label{fig:RestrictedLargeNetworks}
  \end{figure}

\end{example}

Careful interpretation of
the inequalities in
\cref{thm:parameter-conditions-for-belonging-to-m1-intersect-m2} suggests that
the criteria for hybrid distinguishability are highly restrictive, even more
so when biological considerations are taken into account. In particular, we
find three necessary conditions for the hybrid node to be distinguishable.

First, the hybridization must be sufficiently recent. Inequalities
\eqref{eq:solutionCondition-M1-M2} and \eqref{eq:solutionCondition-M1-M3}
establish an exact age cutoff: if the hybridization is sufficiently
ancient, i.e. $a_{1}=e^{-4t_{1}/3}$ is sufficiently small, then the hybrid
node cannot be distinguished, even with unlimited data. Although this cutoff
(the minimum of the right-hand sides of \eqref{eq:solutionCondition-M1-M2} and
\eqref{eq:solutionCondition-M1-M3}) is a complicated function of the model
parameters, it is not difficult to deduce the simpler necessary condition
that distinguishability requires $a_{1}>a_{5}$.
Converting this condition to branch lengths, it holds that if the evolutionary
distance $t_{1}$ between the initial hybrid and its present-day descendants
is greater than the divergence $t_5$ between the hybrid's parent species,
then the hybrid will not be distinguishable.

Second, the evolutionary divergence between the parental species prior to
hybridization must be sufficiently large for the hybrid to be distinguishable.
If too little divergence took place prior to hybridization (i.e., if
$a_{5}=e^{-\frac{4}{3}t_{5}}$ is sufficiently close to $1$), the right-hand
sides of \eqref{eq:solutionCondition-M1-M2} and
\eqref{eq:solutionCondition-M1-M3} exceed $1$, in which case neither
inequality can be satisfied since $a_{1}<1$. This requirement is restrictive
because 3-cycles are expected to occur mostly between closely related taxa, as
hybridization rates decline rapidly with evolutionary divergence
\citep{mallet2007natural,penalba2024role}.

Third, distinguishability requires the hybrid's descendants retain substantial
minor-parent ancestry, i.e. from the parent species with the smaller genetic
contribution to the hybrid. When the proportion of genomic material inherited
from the minor-parent (i.e., $\min\left\{\delta,1-\delta\right\}$) is near
zero, both inequalities \eqref{eq:solutionCondition-M1-M2} and
\eqref{eq:solutionCondition-M1-M3} will fail because their right-hand sides
exceed $1$. This is illustrated in the plot of
\Cref{fig:distinguishability-a5-root}, which shows that distinguishability is
unlikely when $\delta$ is close to $0$ or $1$ (see
\cref{app:details-for-fig:distinguishability-a5-root-simulation}).

In practice, the required combination of high parental divergence together
with substantial minor-parent ancestry is likely to be rare. Although F1
hybrids may exhibit $\delta \approx 1/2$, theoretical and empirical studies
suggest that minor-parent ancestry typically decays rapidly through selective
purging (within tens of generations)
\citep{langdon2024swordtail, moran2021genomic, veller2023recombination}. Thus,
while retained levels of minor-parent ancestry exceeding $10\%$ are not
uncommon \citep{langdon2024swordtail,moran2021genomic}, it is often the case
that minor-parent ancestry is small, and this tends to be incompatible with
distinguishability (see \cref{fig:observation-2-plot,fig:distinguishability-a5-root}). This difficulty
is compounded by the fact that selective purging is driven in part by hybrid
incompatibilities, which increase nonlinearly with parental divergence
\citep{langdon2024swordtail}. Hence, the high parental divergence required for
distinguishability is likely to coincide with \textit{low} minor-parent
ancestry, and this combination is likely to violate the distinguishability
conditions.

\begin{figure*}[]
  \centering

 \includegraphics[scale=.25]{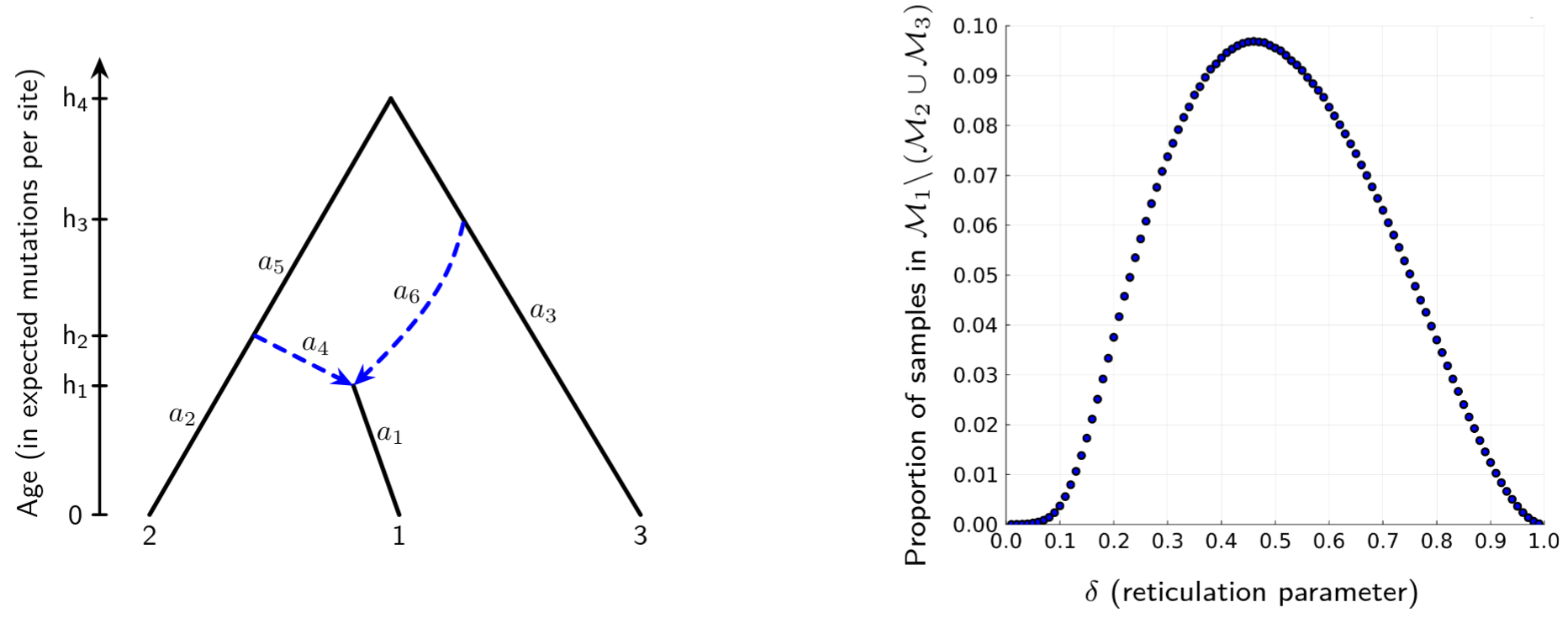}
  \caption{Left: The 3-leaf phylogenetic network $\mathcal{N}_{1}$ with root
    placed on the edge corresponding to the parameter $a_{5}$. Right: The proportion of such
    networks for which the hybrid node is distinguishable, assuming the
    intervals $h_{1}, h_{2}-h_{1}, h_{3}-h_{2}$ and $h_{4}-h_{3}$ are drawn
    uniformly at random from the interval $(0,.5)$. See
    \cref{app:details-for-fig:distinguishability-a5-root-simulation}. }
  \label{fig:distinguishability-a5-root}
\end{figure*}

Finally, we discuss the impact of the root location. The limitations of inferring the hybrid node in a 3-cycle can often, though
not always, be circumvented through additional taxon sampling. If
$\mathcal{N}_{1}$ is rooted on edge $a_{5}$ (as in
\cref{fig:distinguishability-a5-root}), then adding an outgroup will result in
a network with a $4$-cycle, allowing for the use of tests based on $4$-cycles.
On the other hand, if the root lies on one of the leaf edges $2$ or $3$ (as
shown in \cref{fig:ghost-molecular-clock-case}), this need not be possible.
The network in \Cref{fig:ghost-molecular-clock-case} could arise due to
hybridization between two unsampled or extinct ``ghost lineages'' which both
diverged from the ancestor of species $2$ after its divergence with species
$3$. If species $3$ is the closest extant relative of species $1$ and $2$,
then additional taxon sampling will not yield a 4-cycle, so methods specific
to 3-cycles are required.

However, in the ghost-lineage scenario in
\Cref{fig:ghost-molecular-clock-case}, we expect the hybrid node to typically
not be distinguishable from site pattern data, at least when approximately
clock-like evolution is assumed. This is due to our next theorem, which shows
that under a molecular clock model with all leaves equidistant from the root,
the hybrid node will not be distinguishable.

\begin{theorem}
  \label{thm:ghost-molecular-clock-case}
  Under the ghost-lineage scenario shown in
  \Cref{fig:ghost-molecular-clock-case} with the molecular clock assumption,
  for any choice of $0<h_{1}<h_{2}<h_{3}<h_{4}$, it holds that
  $\mathbf{q}\in \mathcal{M}_{1}\cap \mathcal{M}_{2}$. In particular, the
  hybrid node in the 3-cycle is not distinguishable from the site pattern
  distribution.
\end{theorem}

\begin{proof}
  Let $h_{1}<h_{2}<h_{3}<h_{4}$ be arbitrary. By the molecular clock
  assumption, the Fourier parameters are
  \begin{equation}
    \begin{aligned}
      a_{1} &= e^{-\frac{4}{3}h_{1}}
      &\qquad
        a_{2} &= e^{-\frac{4}{3}h_{2}}\\
      a_{3} &= e^{-\frac{4}{3}(2h_{4}-h_{3})}
      &
        a_{4} &= e^{-\frac{4}{3}(h_{2}-h_{1})}\\
      a_{5} &= e^{-\frac{4}{3}(h_{3}-h_{2})}
      &
        a_{6} &= e^{-\frac{4}{3}(h_{3}-h_{1})}.
    \end{aligned}
    \label{eq:molecular-clock-lengths}
  \end{equation}
  We need to show that
  $\mathbf{q}\in \mathcal{M}_{1}\cap \mathcal{M}_{2}$. By
  \cref{eq:solutionCondition-M1-M2} this occurs precisely when
  {\small
    \begin{align*} 
    &a_1(\delta a_4a_5+(1-\delta)a_6)(\delta a_4(1-a_3a_5)+(1-\delta)a_5a_6(1-a_3)) \\
    &< a_5(\delta a_4(1-a_3a_5) + (1-\delta)a_6(1-a_3)).
  \end{align*}}
  By \cref{eq:molecular-clock-lengths}, the Fourier parameters satisfy
  $a_{2}=a_{1}a_{4}$ and $a_{6}=a_{4}a_{5}$. Using these relations, the above
  inequality can be rewritten (code available in the supplemental materials in the file \texttt{ghost-scenario-calculation.m2}) as
  \begin{align*}
    a_6(&a_2a_3a_5^2\delta-a_2a_3a_5^2-a_2a_3a_5\delta -a_2a_5^2\delta\\
    &+a_2a_5^2+a_2\delta+a_3a_5+a_5\delta-a_5-\delta) < 0.
  \end{align*}
  The left-hand side factors as
  \begin{equation*}
    -a_6[\delta (1-a_2)(1-a_3a_5)+(1-\delta )a_5(1-a_3)(1-a_2a_5)],
  \end{equation*}
  which is clearly negative for all $a_{2},a_{3},a_{5},a_{6},\delta\in (0,1)$.
  Hence $\mathbf{q}\in \mathcal{M}_{1}\cap \mathcal{M}_{2}$.
\end{proof}

\begin{figure}
  \centering
  \includegraphics[scale=.25]{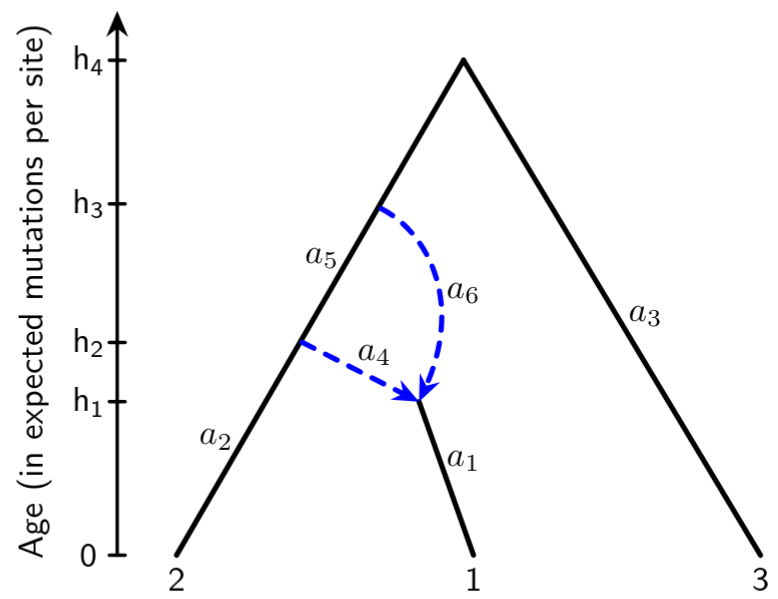}
  \caption{A 3-leaf phylogenetic network rooted on edge $3$ at time $h_{4}$.
    The network exhibits a 3-cycle arising from ghost lineages that diverged
    from species 2 at times $h_{3}$ and $h_{2}$ and later hybridized at time
    $h_{1}$. Note that $a_{3}$ is the Fourier parameter of leaf 3 when the
    root is suppressed.}
  \label{fig:ghost-molecular-clock-case}
\end{figure}

\subsection{Applications to Biological Data}

In this section, we present two examples from the literature where a phylogenetic network is inferred from biological data.
We restrict these networks to a subset of the taxa so that the inferred networks are triangle networks. We then apply the bound from \Cref{thm:parameter-conditions-for-belonging-to-m1-intersect-m2} to determine if the inferred parameters are identifiable under the Jukes-Cantor model. Our goal in this section is simply to give insight into the implications of our results for practical inference. As such, we have derived parameter values from the published literature that are likely of interest to biologists. But, we note that networks were not inferred assuming that the data were generated by the Jukes-Cantor network models we study.

\subsubsection{Melinaea.}

In \citep{vanderHeijden2025}, the authors study the
lineages of \emph{Melinaea} and \emph{Mechanitis} butterflies found throughout
Central and South America and provide evidence of fast and recent radiation.
Here, we examine one of the triangle networks from that paper involving
\emph{Melinaea idae}, \emph{Melinaea lilis}, and \emph{Melinaea marsaeus}.
\emph{Mel. idae} and \emph{Mel. lilis} are sister species while \emph{Mel.
  marsaeus} is a representative of the ingroup clade involved in the study.
The authors fit the network using the
multispecies-coalescent-with-introgression (MSCi) model implemented in BPP
v.4.6.2 \citep{Flouri2019}. For their analysis, they assumed a mutation rate
of $2.9 \times 10^{-9}$ substitutions per site per generation and four
generations per year. Converting the branch lengths from scenario 1 of the
inferred species network in \citep[Fig.~3A]{vanderHeijden2025} into expected
number of substitutions per site, we obtain the network shown in
\Cref{fig:butterfly-example}.

After converting the branch lengths to Fourier parameters, we then substituted
them (and the estimate $\delta=.261$) into the right-hand sides of the
inequalities in
\Cref{thm:parameter-conditions-for-belonging-to-m1-intersect-m2}. This gave
$\approx 1.0004$ for \eqref{eq:solutionCondition-M1-M2} and $\approx 1.0003$
for \eqref{eq:solutionCondition-M1-M3}, both of which are greater than
$a_{1}=.997$. Hence, neither inequality is satisfied, and therefore the hybrid
is \emph{not} distinguishable. Together, these results imply that under the
Jukes-Cantor model with no coalescence, any of the three species could be the
hybrid descendant.

\begin{figure}[h]
  \centering

  \begin{tikzpicture}[scale=1.3, line cap=round, line join=round, font=\footnotesize\sffamily]
    \coordinate (A) at (2.42,0); 
    \coordinate (B) at (2.43,-1); 
    \coordinate (C) at (2.43,1); 
    \coordinate (D) at (0,0); 
    \coordinate (E) at (-1.2,1); 
    \coordinate (F) at (-2.12,0); 
    \coordinate (G) at (-0.01,-1); 
    
    \draw[] (A) -- (D) node[pos=.5, above]{0.00242};
    \draw[] (B) -- (G) node[pos=.5, above]{0.00244};
    \draw[] (C) -- (E) node[pos=.5, above]{0.00363};
    \draw[] (F) -- (E) node[pos=.5, left]{0.00092};
    \draw[] (F) -- (G) node[pos=.5, below left]{0.00211};
    
    \draw[dashed,] (D) -- (E) node[pos=.5, right]{\ 0.0012};
    \draw[dashed] (D) -- (G) node[pos=.45, right]{0.00001};

    \node[right] at (A) {idae};
    \node[right] at (B) {lilis};
    \node[right] at (C) {mars};

  \end{tikzpicture}

  \caption{Network from \cite{vanderHeijden2025} on three species of
    neotropical butterflies with branch lengths converted to expected number
    of mutations per site.}
  \label{fig:butterfly-example}
\end{figure}
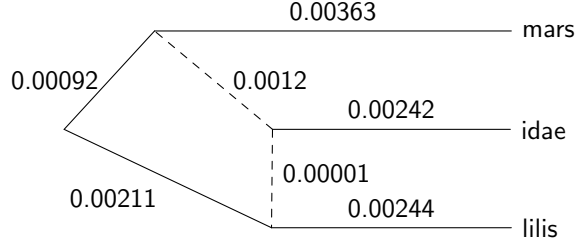

\subsubsection{Rhizoplaca.}

In \citep{keuler2020genome}, the authors use maximum pseudolikelihood,
implemented in \texttt{PhyloNet} \citep{phylonet-paper}, to obtain a
phylogenetic network for seven lichenized fungus species in the genus
\textit{Rhizoplaca}. We replicate their analysis without bootstrap
thresholding, while constraining the number of hybridizations to one. Our code
is available in the supplemental materials in \texttt{src/rhizoplaca/}. The top-scoring network is shown in \Cref{fig:lichen-example};
branch lengths were converted from coalescent units to expected number of
mutations per site using the population mutation parameter
$\theta = 2\times 10^{-3}$, which roughly accords with estimates found by
\cite{leavitt2013local}. This network appears to closely match the network
obtained in \cite[Fig.~4(a)]{keuler2020genome}, which does not include branch
lengths for the leaf edges and reticulation edges since it was inferred from
gene tree topologies.

\begin{figure}[h!]
  \centering
  \includegraphics[scale=.55, trim=2cm 1cm 1.5cm 0.5cm, clip]{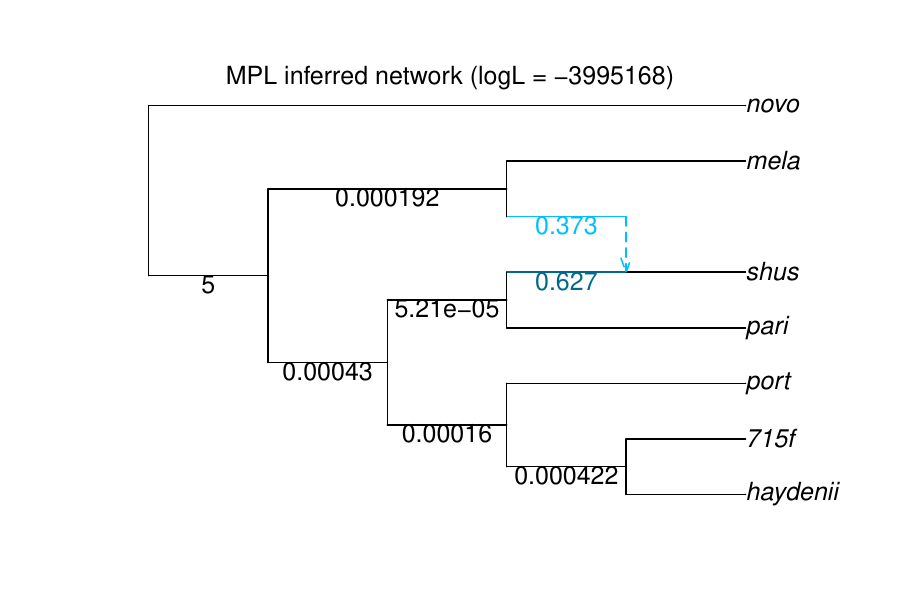}
  \includegraphics[scale=.55, trim=2cm 1cm 1.5cm 1cm, clip]{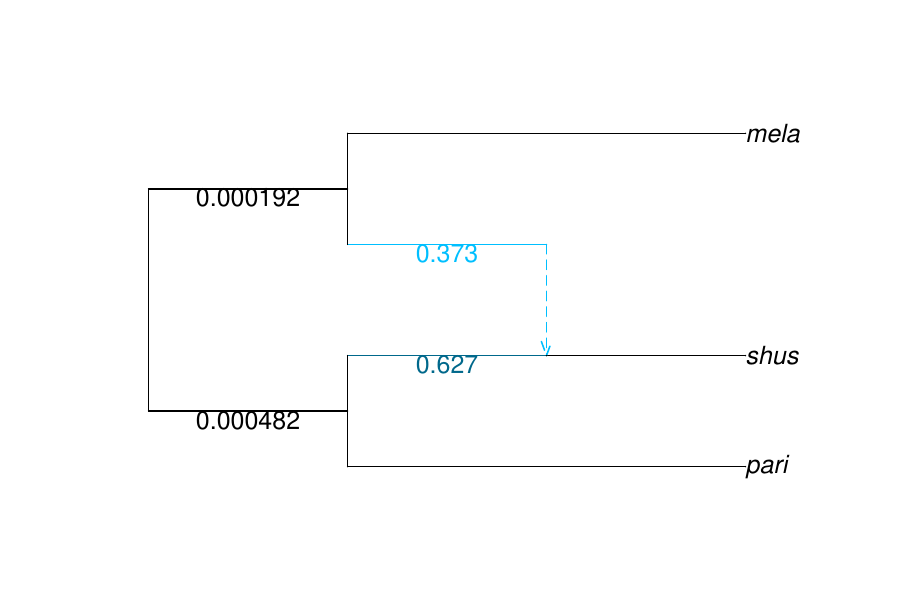}
  \caption{The highest-scoring phylogenetic network for \textit{Rhizoplaca}
    with one reticulation, with internal branch lengths expressed in expected
    number of mutations per site, and its restriction to a 3-taxon network.}
  \label{fig:lichen-example}
\end{figure}

If we restrict the network to the leaves labeled \emph{mela}, \emph{pari}, and
\emph{shus}, we obtain the rooted 3-cycle network with \emph{shus} the hybrid
species (see \cref{fig:lichen-example}). Following
Example \ref{ex:maxexample}, let us suppose that the leaf edges all have
branch length $s$, measured in expected number of mutations per site, and that
the reticulation edges are of negligible length. Using the bound in \Cref{thm:parameter-conditions-for-belonging-to-m1-intersect-m2}, we find that it is possible to determine the direction of hybridization, i.e., that \textit{shus} is the hybrid descendant, if and only if the evolutionary distance $s$ is less than approximately $0.00017$ expected mutations per site.

\section{Discussion}
\label{sec:conclusion}

Since hybridization in eukaryotes occurs most commonly between closely-related
taxa \citep{mallet2016reticulated}, 3-cycles are likely to be common in
nature; however, they have posed significant challenges from an
identifiability standpoint. Due to these challenges, previous identifiability
results have routinely excluded networks with 3-cycles from consideration,
leaving a gap in the literature.

In this work, we fully describe the semialgebraic sets of the Fourier
transformed distributions that arise from each of the three 3-leaf triangle
networks under the Jukes-Cantor model of DNA substitution. Furthermore, for
each pair of 3-leaf triangle networks, we give an explicit semialgebraic
description of the set of parameters of the first model that gets mapped into
the intersection of both models; this set of parameters constitutes
approximately $92.8\%$ of the numerical parameter space
$\Theta_{\mathcal{N}_{1}}\times (0,1)$. Thus, the primary conclusion is that
there is a full-dimensional region in the space of site-pattern probabilities
where identifying the hybrid node is possible, and a full-dimensional region
where identifying the hybrid node is impossible, even with infinitely long
sequences.

Further analysis of our results from a practical lens
suggests that it is generally not possible to determine the orientation of
edges in a 3-leaf triangle network using the site pattern distribution. Hybrid
distinguishability requires recent hybridization, high parental divergence,
and substantial minor-parent ancestry, conditions which are unlikely to all be
satisfied, especially simultaneously. These results are consistent with the
intuition that the signal distinguishing the hybrid lineage is expected to be
weakest when hybridization took place a long time ago, or between
closely-related taxa. What is more surprising is our conclusion, highlighted
in Example \ref{ex:maxexample}, that this loss of identifiability is subject to a
sharp cutoff.

This paper focuses solely on the Jukes-Cantor model, however, we expect that
similar techniques can be used for other group-based models, such as the
Kimura 2-parameter (K2P) and Kimura 3-parameter (K3P) models. The Jukes-Cantor
model is a sub-model of both of these, obtained by identifying certain entries
in the Markov transition matrices. From the model perspective, the
Jukes-Cantor model for a network can be obtained by intersecting the K2P or
K3P model for that network with a linear space. While our work suggests that
similar results might hold for these other models, it is not clear that any
are immediately implied by this work. Thus, it may be the case that hybrid
node identification is more plausible under these different models.

Likewise, the network-based models we study, if we assume ultrametricity, can
be seen as a limit of the network multispecies coalescent (NMSC) model as
population sizes approach zero. Similar identifiability results have also been
established for the NMSC model using semialgebraic conditions.
\cite{allman2024identifiability} show that in level-1 networks, 3-cycles
generally can be detected from gene tree quartet concordance factors under the
NMSC model. In addition, they show that the location of the hybrid node in the
3-cycle can sometimes be determined and sometimes not, depending on the
numerical parameters, the number of leaves, and their position relative to the
hybrid node. Although the results in \citep{allman2024identifiability} do not
imply the results presented in this paper (see Remark 4.7 in
\citep{allman2024identifiability}), there are parallels between their results
and those of this paper. The precise nature of this connection is not
well-understood as the results are based on fundamentally different data
types. A related question, and an especially important one for applications,
is to understand under what conditions 3-cycles can be inferred assuming a
model combining both coalescence and substitution (e.g. see
\cite{allman2022identifiability}).

We have shown that the inequalities of
\Cref{thm:m1-semialgebraic-description-in-qs} contain important information
about model membership. This naturally suggests that it may be possible to use
these inequalities for network inference, in very restrictive settings, using
techniques designed for models described by algebraic constraints, for
example, by using incomplete U-statistics \citep{barnhill2025methodological,
  SturmaNilsDrton} or by substituting sequence data into the defining
polynomial equalities and inequalities and using scoring or supervised machine
learning algorithms for inference \citep{barton2022statistical,
  Casanellas2021SAQ, martin2023algebraic}. However, the results of this paper
point to an important caution in the development of such classification
algorithms. In particular, these algorithms should take into account that the
space of expected site pattern frequencies contains substantial regions where
the classification problem does not have a unique solution. On the other hand,
the results of this paper establish that classification is theoretically
possible over very specific parameter regimes.

Finally, we emphasize that the results of this work should not be interpreted
as settling the question of indistinguishability of 3-cycles more broadly.
Rather, as discussed above, our conclusions are restricted to the problem of
inferring hybrid nodes from site-pattern frequency distributions under the
Jukes–Cantor substitution model. Other sources of genetic signal, including
synteny and gene tree discordance, may provide substantially stronger
information for reconstructing 3-cycles.

\section*{Acknowledgments}

This research began during the Fall 2025 Graduate Research Community at the
University of Hawai`i at M\={a}noa supported by the National Science
Foundation (DGE-2429967). AE was additionally supported by NSF Award
DMS-2023239. EG was additionally supported by DMS-2527518. 

\section*{Supplementary Materials}

Supplementary materials can be found in the following GitHub repository:
\url{https://github.com/max-hill/semialgebraic-conditions-for-triangle-networks}.
\vfill\eject






\begin{thebibliography}{42}
  \providecommand{\natexlab}[1]{#1}
  \providecommand{\selectlanguage}[1]{\relax}
  \providecommand{\bibAnnoteFile}[1]{%
    \IfFileExists{#1}{\begin{quotation}\noindent\textsc{Key:} #1\\
        \textsc{Annotation:}\ \input{#1}\end{quotation}}{}}
  \providecommand{\bibAnnote}[2]{%
    \begin{quotation}\noindent\textsc{Key:} #1\\
      \textsc{Annotation:}\ #2\end{quotation}}

  \bibitem[{Allman et~al.(2025)Allman, An{\'e}, Ba{\~n}os, and
    Rhodes}]{allman2025beyond}
  Allman, E.~S., C.~An{\'e}, H.~Ba{\~n}os, and J.~A. Rhodes. 2025. Beyond
  level-1: Identifiability of a class of galled tree-child networks.
  arXiv:2504.21116.
  \input{allman2025beyond}

  \bibitem[{Allman et~al.(2024)Allman, Ba{\~n}os, Garrote-Lopez, and
    Rhodes}]{allman2024identifiability}
  Allman, E.~S., H.~Ba{\~n}os, M.~Garrote-Lopez, and J.~A. Rhodes. 2024.
  Identifiability of level-1 species networks from gene tree quartets. Bull.
  Math. Biol. 86:110.
  \input{allman2024identifiability}

  \bibitem[{Allman et~al.(2022)Allman, Ba{\~n}os, and
    Rhodes}]{allman2022identifiability}
  Allman, E.~S., H.~Ba{\~n}os, and J.~A. Rhodes. 2022. {Identifiability of
    species network topologies from genomic sequences using the logDet distance}.
  J. Math. Biol. 84:35.
  \input{allman2022identifiability}

  \bibitem[{Ba{\~n}os(2019)}]{banos2019identifying}
  Ba{\~n}os, H. 2019. Identifying species network features from gene tree
  quartets under the coalescent model. Bull. Math. Biol. 81:494--534.
  \input{banos2019identifying}

  \bibitem[{Barley et~al.(2022)Barley, Nieto-Montes~de Oca,
    Manr{\'\i}quez-Mor{\'a}n, and Thomson}]{barley2022evolutionary}
  Barley, A.~J., A.~Nieto-Montes~de Oca, N.~L. Manr{\'\i}quez-Mor{\'a}n, and
  R.~C. Thomson. 2022. The evolutionary network of whiptail lizards reveals
  predictable outcomes of hybridization. Science 377:773--777.
  \input{barley2022evolutionary}

  \bibitem[{Barnhill et~al.(2025)Barnhill, Garrote-L{\'o}pez, Gross, Hill, Kagy,
    Rhodes, and Zhang}]{barnhill2025methodological}
  Barnhill, D., M.~Garrote-L{\'o}pez, E.~Gross, M.~Hill, B.~Kagy, J.~A. Rhodes,
  and J.~Z. Zhang. 2025. Methodological considerations for semialgebraic
  hypothesis testing with incomplete {U}-statistics. arXiv:2507.13531 .
  \input{barnhill2025methodological}

  \bibitem[{Barton et~al.(2026)Barton, Gross, Long, and
    Rusinko}]{barton2022statistical}
  Barton, T., E.~Gross, C.~Long, and J.~Rusinko. 2026. Statistical learning with
  phylogenetic network invariants. Bull. Soc. Syst. Biol. 4 no. 1 (2026) 4.
  \input{barton2022statistical}

  \bibitem[{Casanellas et~al.(2021)Casanellas, Fernández-Sánchez, and
    Garrote-López}]{Casanellas2021SAQ}
  Casanellas, M., J.~Fernández-Sánchez, and M.~Garrote-López. 2021. {SAQ:
    Semi-algebraic quartet reconstruction}. IEEE/ACM Trans. Comput. Biol.
  Bioinform. 18:2855--2861.
  \input{Casanellas2021SAQ}

  \bibitem[{Cui et~al.(2013)Cui, Schumer, Kruesi, Walter, Andolfatto, and
    Rosenthal}]{Cui2013-ps}
  Cui, R., M.~Schumer, K.~Kruesi, R.~Walter, P.~Andolfatto, and G.~G. Rosenthal.
  2013. Phylogenomics reveals extensive reticulate evolution in {X}iphophorus
  fishes. Evolution 67:2166--2179.
  \input{Cui2013-ps}

  \bibitem[{Englander et~al.(2025)Englander, Frohn, Gross, Holtgrefe, van Iersel,
    Jones, and Sullivant}]{englander2025identifiability}
  Englander, A.~K., M.~Frohn, E.~Gross, N.~Holtgrefe, L.~van Iersel, M.~Jones,
  and S.~Sullivant. 2025. {Identifiability of phylogenetic level-2 networks
    under the Jukes-Cantor model}. bioRxiv Pages~2025--04.
  \input{englander2025identifiability}

  \bibitem[{Evans and Speed(1993)}]{evans1993invariants}
  Evans, S.~N. and T.~P. Speed. 1993. Invariants of some probability models used
  in phylogenetic inference. Ann. Stat. Pages~355--377.
  \input{evans1993invariants}

  \bibitem[{Flouri et~al.(2019)Flouri, Jiao, Rannala, and Yang}]{Flouri2019}
  Flouri, T., X.~Jiao, B.~Rannala, and Z.~Yang. 2019. A {Bayesian} implementation
  of the multispecies coalescent model with introgression for phylogenomic
  analysis. Mol. Biol. Evol. 37.
  \input{Flouri2019}

  \bibitem[{Gambette et~al.(2017)Gambette, Huber, and Kelk}]{Gambette2017-kf}
  Gambette, P., K.~T. Huber, and S.~Kelk. 2017. On the challenge of
  reconstructing level-1 phylogenetic networks from triplets and clusters. J.
  Math. Biol. 74:1729--1751.
  \input{Gambette2017-kf}

  \bibitem[{Gross and Long(2018)}]{gross2018distinguishing}
  Gross, E. and C.~Long. 2018. Distinguishing phylogenetic networks. SIAM J.
  Appl. Algebra Geom. 2:72--93. \input{gross2018distinguishing}

  \bibitem[{Gross et~al.(2021)Gross, van Iersel, Janssen, Jones, Long, and
    Murakami}]{gross2021distinguishing}
  Gross, E., L.~van Iersel, R.~Janssen, M.~Jones, C.~Long, and Y.~Murakami. 2021.
  {Distinguishing level-1 phylogenetic networks on the basis of data generated
    by Markov processes}. J. Math. Biol. 83:1--24.
  \input{gross2021distinguishing}

  \bibitem[{Hibbins and Hahn(2022)}]{hibbins2022phylogenomic}
  Hibbins, M.~S. and M.~W. Hahn. 2022. Phylogenomic approaches to detecting and
  characterizing introgression. Genetics 220:iyab173.
  \input{hibbins2022phylogenomic}

  \bibitem[{Holtgrefe et~al.(2025)Holtgrefe, Allman, Ba{\~n}os, van Iersel,
    Moulton, Rhodes, and Wicke}]{holtgrefe2025distinguishing}
  Holtgrefe, N., E.~S. Allman, H.~Ba{\~n}os, L.~van Iersel, V.~Moulton, J.~A.
  Rhodes, and K.~Wicke. 2025. Distinguishing phylogenetic level-2 networks with
  quartets and inter-taxon quartet distances. arXiv:2507.17308.
  \input{holtgrefe2025distinguishing}

  \bibitem[{Jukes et~al.(1969)Jukes, Cantor et~al.}]{jukes1969evolution}
  Jukes, T.~H., C.~R. Cantor, et~al. 1969. Evolution of protein molecules.
  Mammalian protein metabolism 3:132.
  \input{jukes1969evolution}

  \bibitem[{Keuler et~al.(2020)Keuler, Garretson, Saunders, Erickson, St.~Andre,
    Grewe, Smith, Lumbsch, Huang, St.~Clair, and Leavitt}]{keuler2020genome}
  Keuler, R., A.~Garretson, T.~Saunders, R.~J. Erickson, N.~St.~Andre, F.~Grewe,
  H.~Smith, H.~T. Lumbsch, J.-P. Huang, L.~L. St.~Clair, and S.~D. Leavitt.
  2020. Genome-scale data reveal the role of hybridization in lichen-forming
  fungi. Sci. Rep. 10:1497.
  \input{keuler2020genome}

  \bibitem[{Krantz and Parks(2008)}]{krantz2008geometric}
  Krantz, S. and H.~Parks. 2008. Geometric integration theory. Springer.
  \input{krantz2008geometric}

  \bibitem[{Langdon et~al.(2024)Langdon, Groh, Aguillon, Powell, Gunn, Payne,
    Baczenas, Donny, Dodge, Du et~al.}]{langdon2024swordtail}
  Langdon, Q.~K., J.~S. Groh, S.~M. Aguillon, D.~L. Powell, T.~Gunn, C.~Payne,
  J.~J. Baczenas, A.~Donny, T.~O. Dodge, K.~Du, et~al. 2024. Swordtail fish
  hybrids reveal that genome evolution is surprisingly predictable after
  initial hybridization. PLoS Biology 22:e3002742.
  \input{langdon2024swordtail}

  \bibitem[{Leavitt et~al.(2013)Leavitt, Fern{\'a}ndez-Mendoza, P{\'e}rez-Ortega,
    Sohrabi, Divakar, Vondr{\'a}k, Thorsten~Lumbsch, and
    Clair}]{leavitt2013local}
  Leavitt, S.~D., F.~Fern{\'a}ndez-Mendoza, S.~P{\'e}rez-Ortega, M.~Sohrabi,
  P.~K. Divakar, J.~Vondr{\'a}k, H.~Thorsten~Lumbsch, and L.~L.~S. Clair. 2013.
  Local representation of global diversity in a cosmopolitan lichen-forming
  fungal species complex ({Rhizoplaca}, {Ascomycota}). J. Biogeogr.
  40:1792--1806.
  \input{leavitt2013local}

  \bibitem[{Mallet(2005)}]{mallet2005hybridization}
  Mallet, J. 2005. Hybridization as an invasion of the genome. Trends Ecol. Evol. 20:229--237.
  \input{mallet2005hybridization}

  \bibitem[{Mallet et~al.(2007)Mallet, Beltr{\'a}n, Neukirchen, and
    Linares}]{mallet2007natural}
  Mallet, J., M.~Beltr{\'a}n, W.~Neukirchen, and M.~Linares. 2007. Natural
  hybridization in heliconiine butterflies: the species boundary as a
  continuum. BMC Evol. Biol. 7:28.
  \input{mallet2007natural}

  \bibitem[{Mallet et~al.(2016)Mallet, Besansky, and
    Hahn}]{mallet2016reticulated}
  Mallet, J., N.~Besansky, and M.~W. Hahn. 2016. How reticulated are species?
  BioEssays 38:140--149.
  \input{mallet2016reticulated}

  \bibitem[{Martin et~al.(2025)Martin, Holtgrefe, Moulton, and
    Leggett}]{martin2023algebraic}
  Martin, S., N.~Holtgrefe, V.~Moulton, and R.~M. Leggett. 2025. Algebraic
  invariants for inferring 4-leaf semi-directed phylogenetic networks.
  Syst. Biol.  Page~syaf071.
  \input{martin2023algebraic}

  \bibitem[{Moran et~al.(2021)Moran, Payne, Langdon, Powell, Brandvain, and
    Schumer}]{moran2021genomic}
  Moran, B.~M., C.~Payne, Q.~Langdon, D.~L. Powell, Y.~Brandvain, and M.~Schumer.
  2021. The genomic consequences of hybridization. Elife 10:e69016.
  \input{moran2021genomic}

  \bibitem[{Pardi and Scornavacca(2015)}]{Pardi2015-ix}
  Pardi, F. and C.~Scornavacca. 2015. Reconstructible phylogenetic networks: do
  not distinguish the indistinguishable. PLoS Comput. Biol. 11:e1004135.
  \input{Pardi2015-ix}

  \bibitem[{Pe{\~n}alba et~al.(2024)Pe{\~n}alba, Runemark, Meier, Singh, Wogan,
    S{\'a}nchez-Guill{\'e}n, Mallet, Rometsch, Menon, Seehausen
    et~al.}]{penalba2024role}
  Pe{\~n}alba, J.~V., A.~Runemark, J.~I. Meier, P.~Singh, G.~O. Wogan,
  R.~S{\'a}nchez-Guill{\'e}n, J.~Mallet, S.~J. Rometsch, M.~Menon,
  O.~Seehausen, et~al. 2024. The role of hybridization in species formation and
  persistence. Cold Spring Harb. Perspect. Biol. 16:a041445.
  \input{penalba2024role}

  \bibitem[{Rhodes et~al.(2025)Rhodes, Ba{\~n}os, Xu, and
    An{\'e}}]{rhodes2025identifying}
  Rhodes, J.~A., H.~Ba{\~n}os, J.~Xu, and C.~An{\'e}. 2025. Identifying circular
  orders for blobs in phylogenetic networks. Adv. Appl. Math.
  163:102804.
  \input{rhodes2025identifying}

  \bibitem[{Rose et~al.(2025)Rose, Li, Sporck-Koehler, Stacy, Wood, Lemmon,
    Lemmon, Ané, Sytsma, and Givnish}]{lobeloids}
  Rose, J.~P., B.~Li, M.~J. Sporck-Koehler, E.~A. Stacy, K.~R. Wood, E.~M.
  Lemmon, A.~R. Lemmon, C.~Ané, K.~J. Sytsma, and T.~J. Givnish. 2025.
  Phylogenomics of the tetraploid {Hawaiian} lobeliads: Implications for their
  origin, dispersal history, and adaptive radiation. Proc. Natl. Acad. Sci. U.S.A. 122:e2421004122.
  \input{lobeloids}

  \bibitem[{Semple et~al.(2003)Semple, Steel et~al.}]{semple2003phylogenetics}
  Semple, C., M.~Steel, et~al. 2003. Phylogenetics vol.~24. Oxford University
  Press.
  \input{semple2003phylogenetics}

  \bibitem[{Solís-Lemus and Ané(2016)}]{solislemus2016snaq}
  Solís-Lemus, C. and C.~Ané. 2016. Inferring phylogenetic networks with
  maximum pseudolikelihood under incomplete lineage sorting. PLoS Genet.
  12:e1005896.
  \input{solislemus2016snaq}

  \bibitem[{Sturma et~al.(2024)Sturma, Drton, and Leung}]{SturmaNilsDrton}
  Sturma, N., M.~Drton, and D.~Leung. 2024. Testing many constraints in possibly
  irregular models using incomplete {U}-statistics. J. R. Stat. Soc. Ser. B
  Stat Methodol. 86:987--1012.
  \input{SturmaNilsDrton}

  \bibitem[{Sturmfels and Sullivant(2005)}]{SS05}
  Sturmfels, B. and S.~Sullivant. 2005. Toric ideals of phylogenetic invariants.
  J. Comput. Biol. 12:457--481.
  \input{SS05}

  \bibitem[{Sullivant(2023)}]{sullivant2023algebraic}
  Sullivant, S. 2023. {Algebraic Statistics} vol. 194. AMS.
  \input{sullivant2023algebraic}

  \bibitem[{{The algebraic-phylogenetics collaboration}(2026)}]{smalltrees}
  {The algebraic-phylogenetics collaboration}. 2026. A database of small trees
  and networks in algebraic phylogenetics. Version 0.3. Available at
  \url{http://www.algebraicphylogenetics.org}.
  \input{smalltrees}

  \bibitem[{van~der Heijden et~al.(2025)van~der Heijden, Näsvall, Seixas, Nobre,
    Maia, Salazar-Carrión, Walker, Szczerbowski, Schulz, Warren, Córdova,
    Sánchez-Carvajal, Chandi, Arias-Cruz, Rueda-M, Salazar, Dasmahapatra,
    Montgomery, McClure, Absolon, Mathers, Santos, McCarthy, Wood, Lamas,
    Bacquet, Freitas, Willmott, Jiggins, Elias, and Meier}]{vanderHeijden2025}
  van~der Heijden, E. S.~M., K.~Näsvall, F.~A. Seixas, C.~E.~B. Nobre, A.~C.~D.
  Maia, P.~Salazar-Carrión, J.~M. Walker, D.~Szczerbowski, S.~Schulz, I.~A.
  Warren, K.~G.~G. Córdova, M.~J. Sánchez-Carvajal, F.~Chandi, A.~P.
  Arias-Cruz, N.~Rueda-M, C.~Salazar, K.~K. Dasmahapatra, S.~H. Montgomery,
  M.~McClure, D.~E. Absolon, T.~C. Mathers, C.~A. Santos, S.~McCarthy, J.~M.~D.
  Wood, G.~Lamas, C.~Bacquet, A.~V.~L. Freitas, K.~R. Willmott, C.~D. Jiggins,
  M.~Elias, and J.~I. Meier. 2025. Genomics of neotropical biodiversity
  indicators: Two butterfly radiations with rampant chromosomal rearrangements
  and hybridization. Proc. Natl. Acad. Sci. U.S.A.
  122:e2410939122.
  \input{vanderHeijden2025}

  \bibitem[{Veller et~al.(2023)Veller, Edelman, Muralidhar, and
    Nowak}]{veller2023recombination}
  Veller, C., N.~B. Edelman, P.~Muralidhar, and M.~A. Nowak. 2023. Recombination
  and selection against introgressed {DNA}. Evolution 77:1131--1144.
  \input{veller2023recombination}

  \bibitem[{Wen et~al.(2018)Wen, Yu, Zhu, and Nakhleh}]{phylonet-paper}
  Wen, D., Y.~Yu, J.~Zhu, and L.~Nakhleh. 2018. Inferring phylogenetic networks
  using {PhyloNet}. Syst. Biol. 67:735--740.
  \input{phylonet-paper}

  \bibitem[{Xu and An{\'e}(2023)}]{xu2023identifiability}
  Xu, J. and C.~An{\'e}. 2023. Identifiability of local and global features of
  phylogenetic networks from average distances. J. Math. Biol.
  86:12.
  \input{xu2023identifiability}

  \bibitem[{Zhang et~al.(2018)Zhang, Ogilvie, Drummond, and
    Stadler}]{Zhang2018-zk}
  Zhang, C., H.~A. Ogilvie, A.~J. Drummond, and T.~Stadler. 2018. Bayesian
  inference of species networks from multilocus sequence data. Mol. Biol. Evol.
  35:504--517.
  \input{Zhang2018-zk}

\end{thebibliography}

\vfill\eject



\vfill\eject

\appendix

\section{Proof of  \cref{thm:m1-semialgebraic-description-in-qs}}\label{app:proof-of-main-thm}

We begin by introducing a reparameterization which will simplify our analysis.
We substitute $a_4^*:=\delta a_1a_4$ and $a_6^*:=(1-\delta)a_1a_6$ where the
new parameters $a_4^*, a_6^*\in (0,1)$ and must satisfy $a_4^* + a_6^* < 1$,
but are otherwise free. In these new parameters (omitting the stars to
simplify notation), we now have:
\begin{equation}
  \label{lem:allParams}
  \begin{aligned}
    \mathcal{M}_1:\quad 
    &\begin{aligned}
       q_{\mathtt{ACC}} &= a_2 a_3 a_5 
       &\quad q_{\mathtt{CAC}} &= a_3 a_4 a_5 + a_3 a_6 \\
       q_{\mathtt{CCA}} &= a_2 a_4 + a_2 a_5 a_6 
       &\quad q_{\mathtt{CGT}} &= a_2 a_3 a_5 (a_4 + a_6).
     \end{aligned}
  \end{aligned}
\end{equation}

The corresponding parameter space for network $\mathcal{N}_{1}$ is
\begin{gather*} 
  \widetilde{\Theta}_{\mathcal N_1} = \left\{(a_{2},a_{3},a_{4},a_{5},a_{6}):
    0<a_{2},a_{3},a_{4},a_{5},a_{6}< 1 
    \text{ and }a_{4}+a_{6}<1 \right\} \subseteq \mathbb R^5.
\end{gather*}

We will use $\widetilde{\psi}_{\mathcal N_1}$
to denote the new parameterization map shown in \cref{lem:allParams}. With this new notation 
\begin{equation*}
  \mathcal{M}_{1} = \left\{\widetilde{\psi}_{\mathcal N _1}(\theta) : \theta\in  \widetilde{\Theta}_{\mathcal N_1}\right\}.
\end{equation*}
Similar simplifications give reparameterizations for the other two networks $\mathcal{N}_2$ and $\mathcal{N}_3$ shown in \Cref{fig:models-N1-N2-N3-with-labels}, which are as follows:
\begin{equation*}
  \begin{aligned}
    \mathcal{M}_2:\quad 
    &\begin{aligned}
       q_{\mathtt{ACC}} &= b_3 b_4 b_5 + b_3 b_6 
       &\quad q_{\mathtt{CAC}} &= b_1 b_3 b_5 \\
       q_{\mathtt{CCA}} &= b_1 b_4 + b_1 b_5 b_6 
       &\quad q_{\mathtt{CGT}} &= b_1 b_3 b_5 (b_4 + b_6)
     \end{aligned}
    \\[1em]
    \mathcal{M}_3:\quad 
    &\begin{aligned}
       q_{\mathtt{ACC}} &= c_2 c_4 c_5 + c_2 c_6 
       &\quad q_{\mathtt{CAC}} &= c_1 c_4 + c_1 c_5 c_6 \\
       q_{\mathtt{CCA}} &= c_1 c_2 c_5 
       &\quad q_{\mathtt{CGT}} &= c_1 c_2 c_5 (c_4 + c_6).
     \end{aligned}
  \end{aligned}
\end{equation*}

We now prove \cref{thm:m1-semialgebraic-description-in-qs}, which provides a complete semialgebraic description of $\mathcal{M}_1$.  We present the proof of this theorem in the following two subsections. In particular, the equivalence statement in the theorem will follow immediately from \cref{lem:main-theorem--proof-of-forward-direction,lem:main-theorem--proof-of-reverse-direction}, and the additional inequalities \eqref{ineq:CAC-CGT}, \eqref{ineq:CCA-CGT}, and \eqref{ineq:CGT-ACC.CCA} will follow from \cref{lem:additional-inequalities}.

\subsection{Proof of the ``Only If'' Direction of \cref{thm:m1-semialgebraic-description-in-qs}}

Our first lemma establishes the forward direction of \Cref{thm:m1-semialgebraic-description-in-qs}. 

\begin{lemma}\label{lem:main-theorem--proof-of-forward-direction}
  If $(q_{\mathtt{ACC}},q_{\mathtt{CAC}},q_{\mathtt{CCA}},q_{\mathtt{CGT}})\in \mathcal{M}_{1}$ then 
  inequalities \eqref{ineq:0<q<1}, \eqref{ineq:ACC-CGT}, \eqref{ineq:CGT-ACC.CAC}, \eqref{ineq:hybrid}, and \eqref{ineq:a23} hold.
\end{lemma}

\begin{proof}
  By definition of $\mathcal{M}_{1}$, there exists $(a_{2},\ldots,a_{6})\in
  (0,1)^5$  with $a_{4}+a_{6}<1$ such that
  \begin{equation*}
    \widetilde{\psi}_{\mathcal N_1}(a_{2},a_{3},a_{4},a_{5},a_{6}) = (q_{\mathtt{ACC}},q_{\mathtt{CAC}},q_{\mathtt{CCA}},q_{\mathtt{CGT}}).
  \end{equation*}
  Therefore, $q_{\mathtt{ACC}},q_{\mathtt{CAC}},q_{\mathtt{CCA}},$ and $q_{\mathtt{CGT}}$ are given by the formulas in \cref{lem:allParams}, which we will use to prove inequalities \eqref{ineq:0<q<1}, \eqref{ineq:ACC-CGT}, \eqref{ineq:CGT-ACC.CAC}, \eqref{ineq:hybrid}, and \eqref{ineq:a23}.
  
  We begin by proving inequality \eqref{ineq:0<q<1}. By the formulas in \cref{lem:allParams}, it is easily seen that
  $q_{\mathtt{ACC}},q_{\mathtt{CAC}},q_{\mathtt{CCA}},q_{\mathtt{CGT}}>0$ since $a_{2},a_{3},a_{4},a_{5},a_{6}>0$.
  It is also clear that $q_{\mathtt{ACC}},q_{\mathtt{CGT}}<1$ since
  $a_{2},a_{3},a_{4},a_{5},a_{6}<1$ and $a_{4}+a_{6}<1$. It remains to show that
  $q_{\mathtt{CAC}}<1$ and $q_{\mathtt{CCA}}<1$. Since the two are similar, we show only the case of
  $q_{\mathtt{CAC}}$, which proceeds as follows:
  \begin{align*}
    q_{\mathtt{CAC}} 
    &= a_{3}a_{4}a_{5}+a_{3}a_{6} \\
    &= a_{3}(a_{4}a_{5}+a_{6}) \\
    &< a_{3}(a_{4}+a_{6}) &&\text{since $a_{5}<1$}\\
    &< 1,
  \end{align*}
  where the last inequality is due to $a_{3}<1$ and $a_{4}+a_{6}<1$.
  
  Next we will prove inequality \eqref{ineq:ACC-CGT}. By definition of $\widetilde{\psi}_{\mathcal N_1}$,
  \begin{align*}
    q_{\mathtt{ACC}} - q_{\mathtt{CGT}} 
    &= a_2a_3a_5 - (a_2a_3a_4a_5+a_2a_3a_5a_6)  \\
    &= a_2a_3a_5(1-(a_4+a_6))\\
    &> 0,
  \end{align*}
  where the last inequality follows due to $a_{2},a_{3},a_{5}>0$ and
  $a_{4}+a_{6}<1$.

  To prove inequality \eqref{ineq:CGT-ACC.CAC}, first observe that the formulas in
  \cref{lem:allParams} imply
  \begin{align*}
    q_{\mathtt{CGT}} - q_{\mathtt{ACC}}q_{\mathtt{CAC}} 
    &= a_2a_3a_5(a_4+a_6) - a_2a_3^2a_5(a_4a_5+a_6) \\
    &= a_2a_3a_5(a_4(1-a_3a_5) + a_6(1-a_3)).
  \end{align*}
  Again, this must be greater than $0$ since $a_{2},a_{3},a_{4},a_{5},a_{6}>0$ and
  $a_{3},a_{5}<1$.

  Inequality \eqref{ineq:hybrid} was first shown in
  \cite{englander2025identifiability}; the proof presented here is essentially the same as theirs: 
  substituting in the formulas from \cref{lem:allParams} and rearranging terms, one finds that
  \begin{align}\label{eq:hybrid-inequality-in-fourier-parameters}
    q_{\mathtt{ACC}}q_{\mathtt{CAC}}q_{\mathtt{CCA}}-q_{\mathtt{CGT}}^{2}
    &= a_{2}^{2}a_{3}^{2}a_{4}a_{5}a_{6}(1-a_{5})^{2},
  \end{align}
  which is positive since $a_{4},a_{5},a_{6}>0$ and $a_{5}<1$.

  Finally, to prove inequality \eqref{ineq:a23}, observe that by \cref{lem:allParams}, 
  the left hand side of \eqref{ineq:a23} can be written as
  \begin{align*}
    a_2a_3a_5(a_4(1-a_2)(1-a_3a_5)+a_6(1-a_3)(1-a_2a_5)),
  \end{align*}
  which is positive since $a_{2},a_{3},a_{4},a_{5},a_{6}>0$ and $a_{2},a_{3},a_{5}<1$.
\end{proof}

\begin{lemma}[Additional inequalities]
  \label{lem:additional-inequalities}
  If $(q_{\mathtt{ACC}},q_{\mathtt{CAC}},q_{\mathtt{CCA}},q_{\mathtt{CGT}})\in \mathcal{M}_{1}$, then
  \begin{align}
    q_{\mathtt{CAC}} -   q_{\mathtt{CGT}} & >0  \tag{\ref{ineq:CAC-CGT}}\\ 
    q_{\mathtt{CCA}} -   q_{\mathtt{CGT}} & >0  \tag{\ref{ineq:CCA-CGT}}\\ 
    q_{\mathtt{CGT}} - q_{\mathtt{ACC}}q_{\mathtt{CCA}} &> 0  \tag{\ref{ineq:CGT-ACC.CCA}}
  \end{align}
\end{lemma}
\begin{proof}
  Using the formulas from \cref{lem:allParams}, it is easily verified that the
  three inequalities may be written as
  \begin{align*}
    q_{\mathtt{CAC}}-q_{\mathtt{CGT}} &= a_{3}\left[ a_{4}a_{5}(1-a_{2})+a_{6}(1-a_{2}a_{5}) \right] \\
    q_{\mathtt{CCA}}-q_{\mathtt{CGT}} &= a_{2}\left[ a_{4}(1-a_{3}a_{5})+a_{5}a_{6}(1-a_{3}) \right] \\
    q_{\mathtt{CGT}}-q_{\mathtt{ACC}}q_{\mathtt{CCA}} &= a_{2}a_{3}a_{5} \left[ a_{4}(1-a_{2})+a_{6}(1-a_{2}a_{5}) \right].
  \end{align*}
  In all three cases, the right hand sides are seen to be positive by
  inspection, since $0<a_{2},a_{3},a_{4},a_{5}, a_{6}<1$.
\end{proof}

\subsection{Proof of the ``If'' Direction of \cref{thm:m1-semialgebraic-description-in-qs}}
In this section we prove the backwards direction of 
\cref{thm:m1-semialgebraic-description-in-qs}. Namely, we will show that if
$q_{\mathtt{ACC}},q_{\mathtt{CAC}},q_{\mathtt{CCA}}$, and $q_{\mathtt{CGT}}$ satisfy inequalities \eqref{ineq:0<q<1}, \eqref{ineq:ACC-CGT}, \eqref{ineq:CGT-ACC.CAC}, \eqref{ineq:hybrid}, and \eqref{ineq:a23}, 
then there exists $(a_2,a_3,a_4,a_5,a_6)\in \widetilde{\Theta}_{\mathcal N_1}$ such that
$\widetilde{\psi}_{\mathcal N_1}(a_2,a_3,a_4,a_5,a_6) = (q_{\mathtt{ACC}},q_{\mathtt{CAC}},q_{\mathtt{CCA}},q_{\mathtt{CGT}})$.

We begin by showing that for any $\textbf{q} = (q_{\mathtt{ACC}},q_{\mathtt{CAC}},q_{\mathtt{CCA}},q_{\mathtt{CGT}})\in \mathbb{R}^{4}$ satisfying inequality \eqref{ineq:hybrid}, the fiber $\widetilde{\psi}_{\mathcal{N}_{1}}^{-1}(\textbf{q})$ contains a one-dimensional family of points in $\mathbb{R}^5$ parameterized by a single free parameter.

\begin{lemma}[Existence of a one-parameter real family of solutions]
  \label{lem:existence-of-a-family-of-solutions}
  Assume that the coordinates $(q_{\mathtt{ACC}},q_{\mathtt{CAC}},q_{\mathtt{CCA}},q_{\mathtt{CGT}})\in \mathbb{R}^4$ satisfy inequality
  \eqref{ineq:hybrid}. Then there exists a 
  family of solutions in $\mathbb{R}^5$ to the polynomial system 
  \begin{equation}\label{eq:form-of-parameterization}
    \left\{ \begin{array}{l@{}l}
        q_{\mathtt{ACC}}=a_{2}a_{3}a_{5} \\
        q_{\mathtt{CAC}}=a_{3}a_{4}a_{5}+a_{3}a_{6} \\
        q_{\mathtt{CCA}}=a_{2}a_{4}+a_{2}a_{5}a_{6} \\
        q_{\mathtt{CGT}}=a_{2}a_{3}a_{4}a_{5}+a_{2}a_{3}a_{5}a_{6}
      \end{array}\right.
  \end{equation}
  taking the form  $(a_{2},a_{3},a_{4},a_{5},a_{6}) \in \mathbb{R}^{5}$, where
  $a_{2}\in \mathbb{R}\backslash \left\{0\right\}$ is a free variable and 
  
    \begin{align}
    a_{3} &= \frac{q_{\mathtt{ACC}}q_{\mathtt{CAC}}a_{2}-q_{\mathtt{ACC}}q_{\mathtt{CGT}}}{q_{\mathtt{CGT}}a_{2}-q_{\mathtt{ACC}}q_{\mathtt{CCA}}} \label{eq:colbys-a3}\\
    a_{4} &= \frac{a_{2}u}{q_{\mathtt{ACC}}q_{\mathtt{CAC}}a_{2}^{2}-2q_{\mathtt{ACC}}q_{\mathtt{CGT}}a_{2}+q_{\mathtt{ACC}}^{2}q_{\mathtt{CCA}}}\label{eq:colbys-a4}\\
    a_{5} &= \frac{q_{\mathtt{CGT}}a_{2}-q_{\mathtt{ACC}}q_{\mathtt{CCA}}}{q_{\mathtt{CAC}}a_{2}^{2}-q_{\mathtt{CGT}}a_{2}} \label{eq:colbys-a5}\\
      a_{6} &= \frac{(q_{\mathtt{CAC}}a_{2}-q_{\mathtt{CGT}})(q_{\mathtt{CGT}}a_{2}-q_{\mathtt{ACC}}q_{\mathtt{CCA}})}{q_{\mathtt{ACC}}q_{\mathtt{CAC}}a_{2}^{2}
              -2q_{\mathtt{ACC}}q_{\mathtt{CGT}}a_{2}+q_{\mathtt{ACC}}^{2}q_{\mathtt{CCA}}},\label{eq:colbys-a6}
  \end{align}
  where $u=q_{\mathtt{ACC}}q_{\mathtt{CAC}}q_{\mathtt{CCA}}-q_{\mathtt{CGT}}^{2}$.
\end{lemma}
\begin{proof}
  The computations to verify that equations
  \cref{eq:colbys-a3}-\eqref{eq:colbys-a6} yield symbolic
  solutions to \eqref{eq:form-of-parameterization} can be found in the file \texttt{one-param-family-computation.m2} in the Supplementary Materials.  It remains only to show that the denominators of
  \cref{eq:colbys-a3}-\eqref{eq:colbys-a6} are non-zero for all
  $a_{2}\in \mathbb{R}\backslash \left\{0\right\}$.

  We first note that since we assume inequality \eqref{ineq:hybrid}, any solution of the polynomial system \ref{eq:form-of-parameterization} must satisfy
  $a_{2},a_{3},a_{4},a_{5},a_{6} \neq 0$ and $a_{5}\neq 1$,  which
  follows immediately by inspection of
  \cref{eq:hybrid-inequality-in-fourier-parameters}.

  To see that the denominators of \cref{eq:colbys-a3,eq:colbys-a5} are nonzero,
  we write them in terms of the Fourier parameters as follows:
  \begin{equation*}
    q_{\mathtt{CGT}}a_{2}-q_{\mathtt{ACC}}q_{\mathtt{CCA}} = a_{2}^{2}a_{3}a_{5}a_{6}(1-a_{5}).
  \end{equation*}
  and
  \begin{equation*}
    q_{\mathtt{CAC}}a_{2}^{2}-q_{\mathtt{CGT}}a_{2} = a^{2}_{2}a_{3}a_{6}(1-a_{5}).
  \end{equation*}
  Both of these are nonzero since $a_{2},a_{3},a_{5},a_{6}\neq 0$ and $a_{5}\neq1$.

  Finally, observe that the denominator of \cref{eq:colbys-a4,eq:colbys-a6} is
  quadratic in the variable $a_{2}$. The discriminant is
  \begin{equation*}
    4q_{\mathtt{ACC}}^{2} \left( q_{\mathtt{CGT}}^{2}-q_{\mathtt{ACC}}q_{\mathtt{CAC}}q_{\mathtt{CCA}} \right),
  \end{equation*}
  or equivalently, since $q_{\mathtt{ACC}}=a_{2}a_{3}a_{5}$,
  \begin{equation*} 
    4(a_{2}a_{3}a_{5})^{2}
      \left( q_{\mathtt{CGT}}^{2}-q_{\mathtt{ACC}}q_{\mathtt{CAC}}q_{\mathtt{CCA}} \right).
    \end{equation*}
    This is negative by inequality \eqref{ineq:hybrid}, and hence the
    quadratic is nonzero for all
    $a_{2}\in \mathbb{R}\backslash\left\{0\right\}$.
\end{proof}

The next lemma establishes the backwards direction
of \cref{thm:m1-semialgebraic-description-in-qs}.
\begin{lemma}
  \label{lem:main-theorem--proof-of-reverse-direction}
  If $(q_{\mathtt{ACC}},q_{\mathtt{CAC}},q_{\mathtt{CCA}},q_{\mathtt{CGT}}) \in \mathbb R^4$  satisfies inequalities \eqref{ineq:0<q<1}, \eqref{ineq:ACC-CGT}, \eqref{ineq:CGT-ACC.CAC}, \eqref{ineq:hybrid}, and \eqref{ineq:a23}, then $(q_{\mathtt{ACC}},q_{\mathtt{CAC}},q_{\mathtt{CCA}},q_{\mathtt{CGT}}) \in \mathcal{M}_{1}$.
\end{lemma}
\begin{proof} By definition of $\mathcal{M}_{1}$, we need to show there exists
  a solution $(a_{2},a_{3},a_{4},a_{5},a_{6})$ to the polynomial system
  \eqref{eq:form-of-parameterization} satisfying
  $0<a_{2},a_{3},a_{4},a_{5},a_{6}<1$ and $a_{4}+a_{6}<1$.
  
  Assume $a_{3},a_{4},a_{5},a_{6}$ are given by
  \cref{eq:colbys-a3,eq:colbys-a4,eq:colbys-a5,eq:colbys-a6}, with $a_{2}$ to be chosen later. By
  \cref{lem:existence-of-a-family-of-solutions}, 
  it will suffice to show that there exists
  $a_{2}\in (0,1)$ such that $0<a_{3},a_{4},a_{5},a_{6}<1$ and $a_{4}+a_{6}<1$. Let
  \begin{equation*}
    \eta := \frac{q_{\mathtt{ACC}}(q_{\mathtt{CCA}}-q_{\mathtt{CGT}})}{q_{\mathtt{CGT}}-q_{\mathtt{ACC}}q_{\mathtt{CAC}}}.
  \end{equation*}
  We will show that $\eta\in (0,1)$, and that
  $0<a_{3},a_{4},a_{5},a_{6}<1$ and $a_{4}+a_{6}<1$ whenever $\eta < a_{2}< 1$. The proof
  proceeds through a series of claims. Our first claim presents a series of
  useful inequalities which, among other things, will verify that
  $\eta\in (0,1)$.
  
  \begin{claim}
    \label{claim:delta-inequalities}
    $0<\frac{q_{\mathtt{CGT}}}{q_{\mathtt{CAC}}}<\frac{q_{\mathtt{ACC}}q_{\mathtt{CCA}}}{q_{\mathtt{CGT}}}<\eta<1$.
  \end{claim}
  \begin{claimproof}
    We prove the inequalities from left to right. The first inequality follows
    trivially by the assumption that $q_{\mathtt{CGT}},q_{\mathtt{CAC}}>0$, and the second
    follows directly by inequality \eqref{ineq:hybrid}.

    To prove the third inequality, observe that
    $q_{\mathtt{ACC}}q_{\mathtt{CAC}}q_{\mathtt{CCA}}>q_{\mathtt{CGT}}^{2}$ by inequality \eqref{ineq:hybrid}.
    Multiplying by $-q_{\mathtt{ACC}}$ and then adding $q_{\mathtt{ACC}}q_{\mathtt{CCA}}q_{\mathtt{CGT}}$ to both
    sides yields
    \begin{align*}
      &q_{\mathtt{ACC}}q_{\mathtt{CCA}}q_{\mathtt{CGT}}-q_{\mathtt{ACC}}^{2}q_{\mathtt{CAC}}q_{\mathtt{CCA}} < q_{\mathtt{ACC}}q_{\mathtt{CCA}}q_{\mathtt{CGT}}-q_{\mathtt{ACC}}q_{\mathtt{CGT}}^{2}.
    \end{align*}
    Factoring both sides gives
    \begin{equation*}
      q_{\mathtt{ACC}}q_{\mathtt{CCA}}(q_{\mathtt{CGT}}-q_{\mathtt{ACC}}q_{\mathtt{CAC}})< q_{\mathtt{ACC}}q_{\mathtt{CGT}}(q_{\mathtt{CCA}}-q_{\mathtt{CGT}}).
    \end{equation*}
    Dividing both sides by $q_{\mathtt{CGT}}$ and $q_{\mathtt{CGT}}-q_{\mathtt{ACC}}q_{\mathtt{CAC}}$ (which inequality \eqref{ineq:CGT-ACC.CAC} guarantees to be positive) implies
    $\frac{q_{\mathtt{ACC}}q_{\mathtt{CCA}}}{q_{\mathtt{CGT}}}<\eta$. Finally, the inequality
    $\eta<1$ is readily verified using inequalities \eqref{ineq:CGT-ACC.CAC}
    and \eqref{ineq:a23}.
    
  \end{claimproof}

  Since $\eta\in (0,1)$, it remains only to show that
  $0<a_{3},a_{4},a_{5},a_{6}<1$ whenever $\eta < a_{2}<1$. The next
  two claims verify that $a_{3},a_{4},a_{5},a_{6}>0$ whenever
  $\eta< a_{2} <1$.

  \begin{claim}
    \label{claim:a3-and-a5-are-positive}
    If $\eta \leq a_{2}< 1$ then $a_{3}>0$ and $a_{5}>0$.
  \end{claim}
  \begin{claimproof}
    By \cref{eq:colbys-a3,eq:colbys-a5},
    \begin{align*} 
      a_{3} =
      \frac{q_{\mathtt{ACC}}q_{\mathtt{CAC}}\left( a_{2}-\frac{q_{\mathtt{CGT}}}{q_{\mathtt{CAC}}}\right) }
               {q_{\mathtt{CGT}}\left( a_{2}-\frac{q_{\mathtt{ACC}}q_{\mathtt{CCA}}}{q_{\mathtt{CGT}}} \right)}
               \qquad\text{and}\qquad
      a_{5} =
      \frac{q_{\mathtt{CGT}}\left( a_{2} - \frac{q_{\mathtt{ACC}}q_{\mathtt{CCA}}}{q_{\mathtt{CGT}}} \right)}
      {q_{\mathtt{CAC}}a_{2}\left(a_{2}-\frac{q_{\mathtt{CGT}}}{q_{\mathtt{CAC}}}\right)}.
    \end{align*}
    Using the inequalities in the statement of \cref{claim:delta-inequalities}, it is
    clear by inspection of these formulas that both $a_{3}$ and $a_{5}$ are
    positive when $a_{2} \geq \eta$.
    
  \end{claimproof}

  \begin{claim}
    \label{claim:a4-and-a6-are-positive}
    If $\eta \leq a_{2}< 1$ then  $a_{4}> 0$ and $a_{6}>0$.
  \end{claim}
  \begin{claimproof}
    Assume $a_{2}\geq \eta$. To show that $a_{4}>0$, we will show that both the numerator and
    denominator of \cref{eq:colbys-a4} are positive. By completing the square, the denominator can be written as
    \begin{align}\label{eq:a4-denominator-completed-square}
      \begin{aligned}
        &q_{\mathtt{ACC}}q_{\mathtt{CAC}}a^{2}_2-2q_{\mathtt{ACC}}q_{\mathtt{CGT}}a_2+q_{\mathtt{ACC}}^{2}q_{\mathtt{CCA}} 
        = q_{\mathtt{ACC}}q_{\mathtt{CAC}} 
          \left( a_{2}-\frac{q_{\mathtt{CGT}}}{q_{\mathtt{CAC}}} \right)^{2}+ \frac{q_{\mathtt{ACC}}}{q_{\mathtt{CAC}}}u,
      \end{aligned}
    \end{align}
    which is positive since $u>0$ due to inequality \eqref{ineq:hybrid}. The numerator is also positive by \eqref{ineq:hybrid}, hence $a_{4}>0$.

    Similarly, we will show $a_{6}>0$ using the same approach, namely by showing that both the numerator and
    denominator of \cref{eq:colbys-a6} are positive. Positivity of the
    denominator follows immediately by
    \cref{eq:a4-denominator-completed-square}. To see that the numerator is
    also positive, first observe that it can be written as
    \begin{align*}
      &(q_{\mathtt{CAC}}a_{2}-q_{\mathtt{CGT}})(q_{\mathtt{CGT}}a_{2}-q_{\mathtt{ACC}}q_{\mathtt{CCA}}) 
      = q_{\mathtt{CAC}}q_{\mathtt{CGT}} \left( a_{2} - \frac{q_{\mathtt{CGT}}}{q_{\mathtt{CAC}}} \right) \left( a_{2}- \frac{q_{\mathtt{ACC}}q_{\mathtt{CCA}}}{q_{\mathtt{CGT}}} \right).
    \end{align*}
    \cref{claim:delta-inequalities} implies that the right-hand side is a product of
    positive terms when $a_{2}\geq \eta$. Hence $a_{6}>0$.
    
  \end{claimproof}

  It remains to show that $a_{3},a_{4},a_{5},a_{6}<1$ whenever $\eta <
  a_{2}< 1$, which we show in the next three claims.

  \begin{claim}
    \label{claim:a3-is-less-than-1}
    If $\eta < a_{2} < 1$ then $a_{3}<1$.
  \end{claim}
  \begin{claimproof}
    Taken together, the inequalities $\eta<a_{2} <1$ and inequality \eqref{ineq:CGT-ACC.CAC}
    imply that
    \begin{equation*}
      q_{\mathtt{ACC}}(q_{\mathtt{CCA}}-q_{\mathtt{CGT}})< (q_{\mathtt{CGT}}-q_{\mathtt{ACC}}q_{\mathtt{CAC}})a_{2}.
    \end{equation*}
    Rearranging terms gives
    \begin{equation*} 
      q_{\mathtt{ACC}}q_{\mathtt{CAC}}a_{2} - q_{\mathtt{ACC}}q_{\mathtt{CGT}}< q_{\mathtt{CGT}}a_{2} - q_{\mathtt{ACC}}q_{\mathtt{CCA}}.
    \end{equation*}
    The right-hand side is positive (by \cref{claim:delta-inequalities} and the
    assumption that $a_{2}> \eta$), and hence
    \begin{equation*}
      \frac{q_{\mathtt{ACC}}q_{\mathtt{CAC}}a_2 - q_{\mathtt{ACC}}q_{\mathtt{CGT}}}{q_{\mathtt{CGT}}a_{2} - q_{\mathtt{ACC}}q_{\mathtt{CCA}}}<1.
    \end{equation*}
    Since the left hand side is precisely $a_3$ by \cref{eq:colbys-a3}, the
    claim follows.
    
  \end{claimproof}

  \begin{claim}
    \label{claim:a4-and-a6-are-less-than-1}
    If $\eta \leq a_{2}< 1$ then $a_{4}<1$, $a_{6}<1$, and $a_{4}+a_{6}<1$.
  \end{claim}
  \begin{claimproof}
    By the formulas in \cref{eq:form-of-parameterization} and inequality
    \eqref{ineq:ACC-CGT}, $a_{4}+a_{6} = \frac{q_{\mathtt{CGT}}}{q_{\mathtt{ACC}}}<1.$ Moreover by
    \cref{claim:a4-and-a6-are-positive}, $a_{4},a_{6}>0$ when $a_{2}\geq \eta$.
    Therefore $a_{4}<1$ and $a_{6}<1$.
    
  \end{claimproof}

  \begin{claim}
    \label{claim:a5-is-less-than-1}
    If $\eta \leq a_{2} < 1$ then $a_{5}<1$.
  \end{claim}
  \begin{claimproof}
    Inequality \eqref{ineq:hybrid} states that $u > 0$, so the following inequality holds for any $a_{2}\in \mathbb{R}$:
    \begin{equation*}
      q_{\mathtt{CAC}} \left( a_{2}-\frac{q_{\mathtt{CGT}}}{q_{\mathtt{CAC}}} \right)^{2} + \frac{u}{q_{\mathtt{CAC}}}>0. 
    \end{equation*}
    Expanding and rearranging terms, we can rewrite the above inequality as
    \begin{equation}\label{eq:a5-is-less-than-1-eq1}
      q_{\mathtt{CGT}}a_{2} - q_{\mathtt{ACC}}q_{\mathtt{CCA}}< q_{\mathtt{CAC}}a_{2}^{2} -q_{\mathtt{CGT}}a_{2}.
    \end{equation}
    If $a_{2} \geq \eta$, then \cref{claim:delta-inequalities} implies
    $a_{2}>q_{\mathtt{CGT}}/q_{\mathtt{CAC}}$, which implies that the right-hand side of
    \cref{eq:a5-is-less-than-1-eq1} is positive. Therefore
    \begin{equation*}
      \frac{q_{\mathtt{CGT}}a_{2} - q_{\mathtt{ACC}}q_{\mathtt{CCA}}}{q_{\mathtt{CAC}}a_{2}^{2} -q_{\mathtt{CGT}}a_{2}}<1,
    \end{equation*}
    and hence $a_{5}<1$ by \cref{eq:colbys-a5}.
    
  \end{claimproof}

  We can now finish the proof of \cref{lem:main-theorem--proof-of-reverse-direction}. By
  \cref{claim:delta-inequalities}, $\eta\in (0,1)$, allowing us to choose $a_{2}\in (\eta,1)$. Letting $a_{3},a_{4},a_{5},a_{6}$ be given
  by the formulas in
  \cref{eq:colbys-a3}-\eqref{eq:colbys-a6},
  \cref{lem:existence-of-a-family-of-solutions} implies that
  $(a_{2},a_{3},a_{4},a_{5},a_{6})$ is a solution to the polynomial system
  \eqref{eq:form-of-parameterization}. Moreover $a_{3},a_{4},a_{5},a_{6}>0$ (by \cref{claim:a3-and-a5-are-positive,claim:a4-and-a6-are-positive}), and
  $a_{3},a_{4},a_{5},a_{6}<1$
  (by \cref{claim:a3-is-less-than-1,claim:a4-and-a6-are-less-than-1,claim:a5-is-less-than-1}).
  Hence $(a_{2},a_{3},a_{4},a_{5},a_{6})\in (0,1)^{5}$. Finally, $a_{4}+a_{6}<1$ by \cref{claim:a4-and-a6-are-less-than-1}.
\end{proof}

Thus, we have proven both directions of \Cref{thm:m1-semialgebraic-description-in-qs}, giving a complete semialgebraic description of $\mathcal{M}_1$.


\section{The Fourier Transform}
\label{app:fourier-transform}
In this section we show how to use the Fourier transform to translate between
the site pattern probabilities defined in
\cref{eq:site-pattern-probability-formula} and the simplified Jukes-Cantor
$q$-coordinates defined in \cref{eq:q-defs}. We also prove the uniform
volume-scaling property cited in \Cref{obs:intersection-volumes-in-qs}.
Accompanying code and computations for this section can be found in the
supplemental materials in the file \texttt{Fourier-transform-appendix.jl}.
For additional details on the Fourier transformation, see \citep{SS05} and \citep[Chapter~15, pp.~335--370]{sullivant2023algebraic}.

We begin by making the observation that a given site-pattern distribution
arising from a $3$-leaf network (defined in
\cref{eq:site-pattern-probability-formula}) may be encoded as a point
$\mathbf{p}:=(p_{0},p_{1},p_{2},p_{3},p_{4})^{\top}\in \Delta_{4}$, where
\begin{align*}
  p_{0} &=\Pr\left[X_{1}=X_{2}=X_{3} \right]\\
  p_{1} &=\Pr\left[X_{1}\neq X_{2}=X_{3} \right]\\
  p_{2} &=\Pr\left[X_{1}=X_{3}\neq X_{2} \right]\\
  p_{3} &=\Pr\left[X_{1}=X_{2}\neq X_{3} \right]\\
  p_{4} &=\Pr\left[X_{1},X_{2},X_{3} \text{ are all distinct} \right].
\end{align*}
Due to symmetries of the Jukes-Cantor model, there are $5$ equivalence classes
for the leaf-state assignment probabilities of
\cref{eq:site-pattern-probability-formula}, and we may take
$p_{\mathtt{AAA}},p_{\mathtt{ACC}},p_{\mathtt{CAC}},p_{\mathtt{CCA}},p_{\mathtt{CGT}}$ as our system of representatives
(see \cite{smalltrees}). Then we observe that since
$p_{\mathtt{AAA}}=p_{\mathtt{CCC}}=p_{\mathtt{GGG}}=p_{\mathtt{TTT}}$,
\begin{equation*}
  p_{0} = p_{\mathtt{AAA}}+p_{\mathtt{CCC}}+p_{\mathtt{TTT}}+p_{\mathtt{GGG}} = 4p_{\mathtt{AAA}}.
\end{equation*}
Similar considerations
yield $p_{1}=12 p_{\mathtt{ACC}}$, $p_{2} = 12p_{\mathtt{CAC}}$, $p_{3} = 12p_{\mathtt{CCA}}$, and
$p_{4} = 24p_{\mathtt{CGT}}$.

To obtain the simplified Jukes-Cantor $q$-coordinates, we apply the following linear
change of coordinates:
\begin{equation*}
  \begin{bmatrix}
    q_{\mathtt{AAA}}\\q_{\mathtt{ACC}}\\q_{\mathtt{CAC}}\\q_{\mathtt{CCA}}\\q_{\mathtt{CGT}}
  \end{bmatrix}
  =
  \begin{bmatrix}
    1&     1&     1&     1&     1\\
    1&     1&  -1/3&  -1/3&  -1/3\\
    1&  -1/3&     1&  -1/3&  -1/3\\
    1&  -1/3&  -1/3&     1&  -1/3\\
    1&  -1/3&  -1/3&  -1/3&  \phantom{-}1/3
  \end{bmatrix}
  \begin{bmatrix}
    p_{0}\\p_{1}\\p_{2}\\p_{3}\\p_{4}
  \end{bmatrix}.
\end{equation*}
The $5\times 5$ matrix in this equation is the \textit{specialized Fourier transformation} \citep{smalltrees}, which we will denote by $F$. Note that $q_{\mathtt{AAA}}=1$ since $p_{0}+p_{1}+p_{2}+p_{3}+p_{4}=1$. Thus,
to obtain the simplified Jukes-Cantor $q$-coordinates from a point
$\mathbf{p}\in \Delta_{4}$, we apply the
transformation $\mathcal{F} := \pi\circ F$,
where $\pi$ is the projection
\begin{equation*}
  \pi(x_{1},x_{2},x_{3},x_{4},x_{5})=(x_{2},x_{3},x_{4},x_{5}).
\end{equation*}
Formally, the \textit{space of simplified Jukes-Cantor q-coordinates} is
the set
$\mathcal{S} := \left\{\mathcal{F}(\mathbf{p}) : \mathbf{p}\in
  \Delta_{4}\right\}$. The site pattern probabilities can be recovered from
the simplified Fourier coordinates by inverting $F$.

The next lemma shows that $\mathcal{F}$ scales 4-dimensional volume uniformly
by the constant factor $\frac{2}{\sqrt{5}}\left( \frac{4}{3} \right)^{4}$.
\begin{lemma}[Volume scaling]
  \label{lem:volume-scaling}
  Let $E\subseteq \Delta_{4}$ be measurable. Then
  \begin{equation*}
    |\mathcal{F}(E)| 
    = \frac{2}{\sqrt{5}}\left(\frac{4}{3}\right)^{4} \cdot |E|,
  \end{equation*}
  where $|\cdot|$ denotes the $4$-dimensional volume.
\end{lemma}

\begin{proof}
  It is straightforward to verify that for all 
  $\mathbf{p}=(p_{0},p_{1},p_{2},p_{3},p_{4})^{\top}\in \Delta_{4}$,
  \begin{equation}\label{eq:2}
    \mathcal{F}(\mathbf{p}) = \mathbf{1} + A \pi(\mathbf{p})
  \end{equation}
  where $\mathbf{1}=(1,1,1,1)^{\top}$ and
  \begin{equation*}
    A= -\frac{4}{3}
    \begin{bmatrix}
      0&1&1&1\\
      1&0&1&1\\
      1&1&0&1\\
      1&1&1&1/2
    \end{bmatrix}.
  \end{equation*}
  A direct computation gives $\det(A) = -2 \left( \frac{4}{3} \right)^{4}$.

  Further, observe that when restricted to $\Delta_{4}$, the map $\pi$ is a bijection
  onto $\pi(\Delta_{4})$, with inverse
  $T(u_{1},\ldots,u_{4})=
  (1-u_{1}-u_{2}-u_{3}-u_{4},u_{1},u_{2},u_{3},u_{4})^{\top}$.  Since the $4$-dimensional Jacobian of $T$ is $\sqrt{5}$, it follows by the
  area formula \cite[Theorem~5.1.1, Chapter 5, p.125]{krantz2008geometric} applied to $\pi^{-1}=T$ that
  \begin{equation}\label{eq:5}
    |\pi(E)| = \frac{1}{\sqrt{5}} |E|.
  \end{equation}
  for every measurable $E\subseteq \Delta_{4}$.
  Therefore, since volume is translation invariant,
  \begin{align*}
    |\mathcal{F}(E)|
    &=|\mathbf{1}+A\pi(E)| \\
    &=|A\pi(E)|\\
    &=|\det(A)|\cdot |\pi(E)|\\
    &=2\left( \frac{4}{3} \right)^{4} \cdot | \pi(E)| &&\text{since $|\det(A)| = 2\left(\frac{4}{3}\right)^{4}$}\\
    &=\frac{2}{\sqrt{5}} \left( \frac{4}{3} \right)^{4}\cdot |E| &&\text{by \cref{eq:5}}.
  \end{align*}
\end{proof}

\begin{remark}
  \label{rmk:volume-scaling}
  \cref{lem:volume-scaling} can be used to justify the volume claims made in
  \cref{obs:intersection-volumes-in-qs}. For example, since
  $|\Delta_{4}| = \frac{\sqrt{5}}{24}$, \cref{lem:volume-scaling} implies that the space of 
  simplified Fourier coordinates $\mathcal{S}$ has 4-dimensional volume
  $4^{3}/3^{5}\approx 0.263$.  Since approximately $1.76\%$ of $\Delta_{4}$ is mapped to
  $\mathcal{M}_{1}$ (as shown in \Cref{fig:venn-simplex}), and since the transformation $\mathcal{F}$ scales the $4$-dimensional volume of subsets of $\Delta_{4}$ uniformly, it follows that $\mathcal{M}_{1}$ occupies about $1.76\%$ of the volume of $\mathcal{S}$. Hence $\mathcal{M}_{1}$ has volume
  $(0.0176)|\mathcal{S}| \approx 0.00464$.
  
  Similarly, \Cref{fig:venn-simplex} shows that for any $i\neq j$, approximately $0.7\%$ of the simplex is mapped to $\mathcal{M}_{i}\cap\mathcal{M}_{j}$, and approximately $0.56\%$ is mapped to $\mathcal{M}_{1}\cap\mathcal{M}_{2}\cap\mathcal{M}_{3}$. Therefore $|\mathcal{M}_{i}\cap\mathcal{M}_{j}|=(0.007)|\mathcal{S}|\approx 0.00184$ for $i\neq j$, and $|\mathcal{M}_{1}\cap\mathcal{M}_{2}\cap\mathcal{M}_{3}|=(0.0056)|\mathcal{S}|\approx 0.00147$.
\end{remark}

\section{Details for \Cref{fig:distinguishability-a5-root} Simulation}
\label{app:details-for-fig:distinguishability-a5-root-simulation}
For each $\delta\in \left\{.01,.02,\ldots,.99\right\}$ we simulated $n=10^{7}$
random networks of the form shown in \Cref{fig:distinguishability-a5-root}. In
particular, $h_{4}$ is the height of the network, and edges $a_{6}$ and
$a_{4}$ diverge at times $h_{3}$ and $h_{2}$ respectively, and then hybridize
at time $h_{1}$. We assumed a molecular clock, so that the Fourier parameters
were taken to be
\begin{align*}
  a_{i}  &=  e^{-\frac{4}{3}h_{i}}, \ i=1,2,3,\\
  a_{4}  &= e^{-\frac{4}{3}(h_{2}-h_{1})}\\
  a_{5}  &= e^{-\frac{4}{3}(2h_{4}-h_{2}-h_{3})}\\
  a_{6}  &= e^{-\frac{4}{3}(h_{3}-h_{1})}
\end{align*}
where the intervals $h_{1},h_{2}-h_{1},h_{3}-h_{2},$ and $h_{4}-h_{3}$ were
drawn independently and uniformly at random from the interval $(0,0.5)$.
Varying the maximum of the uniform distribution had the effect of scaling the
height of the bump, with `proportion distinguishable' increasing when $.5$ was
replaced by $1$ and decreasing when smaller values were used.
Code for running this simulation and replicating the plot in
\Cref{fig:distinguishability-a5-root} can be found in the file
\texttt{delta-figure-plot.jl} in the Supplementary Materials.
\newpage

\end{document}